\documentclass[reqno]{amsart}
\usepackage{amsmath}
\usepackage{amsfonts}
\usepackage{amssymb}
\usepackage{graphicx}
\usepackage{amsthm}

\numberwithin{equation}{section}

\theoremstyle{plain}
\newtheorem{theorem}{Theorem}[section]
\newtheorem{proposition}[theorem]{Proposition}
\newtheorem{lemma}[theorem]{Lemma}
\newtheorem{corollary}[theorem]{Corollary}
\newtheorem{remark}[theorem]{Remark} 
\newtheorem{definition}[theorem]{Definition}

\newtheorem{hypothesis}[theorem]{Assumption}

\renewcommand{\P}{\mathbb{P}}
\newcommand{\E}{\mathbb{E}}
\newcommand{\R}{\mathbb{R}}
\newcommand{\F}{\mathcal{F}}
\newcommand{\FF}{\mathbb{F}}

\title{The Self-Financing Equation in High Frequency Markets}
\date{\today}
\author{Ren\'e Carmona}
\author{Kevin Webster}

\begin{document}

\begin{abstract}
High Frequency Trading (HFT) represents an ever growing proportion of all financial transactions as most markets have now switched to electronic order book systems. The main goal of the paper is to propose continuous time equations which generalize the self-financing relationships of frictionless markets to electronic markets with limit order books. We use NASDAQ ITCH data to identify significant empirical features such as price impact and recovery, rough paths of inventories and vanishing bid-ask spreads. Starting from these features, we identify microscopic identities holding on the trade clock, and through a diffusion limit argument, derive continuous time equations which provide a macroscopic description of properties of the order book.
These equations naturally differentiate between trading via limit and market orders. We give several applications (including hedging European options with limit orders, market maker optimal spread choice, and toxicity indexes) to illustrate their impact and how they can be used to the benefit of Low Frequency Traders (LFTs).
\end{abstract}

\maketitle

\section{\textbf{Introduction}}
In a series of papers (\cite{O'Hara3, O'Hara2, O'Hara}) on the divide between high and low frequency traders, M. O'Hara and co-authors identified a number of market features that both Low Frequency Traders (LFTs for short) and most academic researchers have largely ignored, but that High Frequency Traders (HFTs from now on)  exploit with great success. 

\begin{quote}
``There is no question that the goal of many HFT strategies is to profit from LFTs mistakes. [...] Part of HFTs success is due to the reluctance of LFT to adopt (or even to recognize) their paradigm.''(\cite{O'Hara})
\end{quote}

These papers also outline a program to better understand and possibly remedy these issues: in a nutshell, these authors recommend that LFTs update the strategies and models they use in order to incorporate more of the features of the high frequency markets. While the goal should not be to try to \emph{beat} the HFTs at their own game by modeling the high frequency market microstructure in painstaking detail, it should be to \emph{capture}, at least sparsely, the macroscopic effects of those phenomena that actually affect LFT.

This paper is in line with this  program. Case in point, its main thrust is to provide forms of the \emph{self-financing portfolio equation}, both in discrete and continuous time, consistent with the high frequency paradigm. The  equations we propose are motivated by and fitted to high frequency data. They are  derived theoretically from accounting rules at the high frequency level. Their continuous time limits capture the relevant effects at the macroscopic level. From these fundamental relationships, we use the powerful tools of stochastic calculus to revisit the solutions of a certain number of standard continuous time financial problems in light of the new high frequency paradigm. We show how the latter affects for example option hedging and we highlight the different solution depending upon trading being through limit orders versus market orders. A model for market making in the spirit of \cite{Stoikov} is solved. We also introduce, still in the same framework, an instantaneous and a cumulative toxicity indexes in the spirit of \cite{O'Hara}.

\vskip 4pt
The crucial insight of \cite{O'Hara}, named 'the new paradigm', is the fact that high frequency traders do not operate on the 'calendar' clock, but instead use some form of 'event-based time', such as the trade clock, or the volume clock. This is partly due to the algorithmic nature of their strategies and the lack of direct calendar clock dependent constraints such as maturities and the likes. A fringe benefit for quantitative analysis is the well documented fact that prices \emph{behave better} under an event-based clock than the calendar clock. A number of papers \cite{Ane, Clark, Clark2, Mandelbrot2, Mandelbrot, O'Hara} argue that, in addition to removing seasonal effects and resolving asynchronicity issues, this time-change makes the price returns more Gaussian-like. Even though this property is mostly irrelevant in our analysis, we choose to work in the \emph{trade clock} in which each discrete time step corresponds to one trade. Indeed, even though our conclusions are independent of the clock used, we find the trade clock especially convenient to formulate and test the significance of our findings. With these proviso out of the way, we can outline our research agenda:

\begin{enumerate}
\item Understand, at the microscopic level, structural relationships and strategies that HFTs exploit;\label{enum_micro}
\item Identify which features persist at the macroscopic level, in which form, and provide continuous time models on that scale;\label{enum_macro}
\item Use these models to update LFT strategies and provide monitoring tools: transaction cost analysis, measure of toxicity of order flow, $\ldots$\label{enum_appli}
\end{enumerate}
For the sake of definiteness, we focus on the self-financing portfolio equation of continuous time finance. To this effect, we review in Section \ref{se:sfe} the role of this condition in quantitative finance, and in so doing,  introduce the continuous time analysis notation used in the paper, as well as the exact form of our generalization.

\vskip 2pt
The main originality of the form of the  self financing condition which we propose to use, is the fact that it accounts for both price impact and price recovery, two important empirical microstructure features that are usually ignored or modeled in separate ad-hoc fashions. It also differentiates between the impacts of limit and market orders. This is important because nowadays, a large number of agents  trade with both types of orders, rather than simply relying on market makers to find trades. Furthermore, our generalization of the self-financing portfolio equation can be used with a larger class of inventories models, e.g. \emph{with infinite variation}. This allows the use of the powerful tools of stochastic calculus to retain tractability in a number of models. 

\vskip 6pt
The classical self-financing portfolio equation was generalized in two separate directions in the financial engineering literature.
On one hand, Almgren and Chriss proposed in \cite{Almgren} a way to incorporate price impact and temporary transaction costs in a phenomenological model for optimal execution with market orders and finite speed of trading. On the other hand, and with a completely different point of view, extensions of the classical self-financing equation of the Black-Scholes theory were touted by researchers attempting to include transaction costs in Merton's optimal portfolio's theory. See for example  \cite{ Chellathurai, Magill, Shreve} or the recent review \cite{MuhleKarbe}.

\vskip 12pt
Two books, \emph{Empirical market microstructure} by J. Hasbrouck (\cite{Hasbrouck}) and \emph{Market microstructure theory} by M. O'Hara (\cite{O'Hara4}) cover the state of the field prior to the advent of HFT. They contain informed trader models (\cite{Kyle}) and inventory-based market making (\cite{Amihud, Garman, Stoll, O'Hara5}). \emph{Three} main themes united different market structures at that time: the limit order book, adverse selection (the underlying cause of price impact) and statistical predictions. These themes are just as relevant, if not  more so in the new age of high frequency trading.

\vskip 2pt
Our investigations were inspired by a large number of empirical studies of high frequency data (see for example \cite{Biais, Bouchaud, Bouchaud3, Chakraborti, Maslov2, Weber, Bouchaud2, Farmer}), and recent publications of theoretical models of the limit order book (\cite{Cont3, Cont2, Cont, Maslov, Bouchaud2}). However, our emphasis is different as we use limit orders as a \emph{starting point}. Our goal is not to \emph{explain} the evolution of the order book, but merely to analyze the \emph{consequences} of the choices made by the liquidity providers and takers on price changes, their inventories and their wealth.

\vskip 6pt
We close this introduction with a short overview of the paper. Since so much of our motivation and results depend upon the self financing condition, we devote next section to a review of the role of this condition in continuous time quantitative finance, with the goal of introducing the notation used in the paper, as well as announcing the exact form of our generalization.
The remainder of the paper is structured into two parts. In the first part, we consider limit order books on which the trades take place at the best bid and best ask only. While seemingly restrictive, this assumption can be justified by looking closely at the data. Indeed, once two specific classes of executions are removed from the data \footnote{We removed two specific classes of trades: 1) executions classified by NASDAQ as \emph{type 'C'}. While we were not able to figure out what these special deals are, their numbers are very small, and on any given day, for any given stock, these executions represent less than $1\%$ of the trades; 2) executions of hidden orders. While in very small numbers, if at all present, for small cap stocks, these trades are frequently very significant for large cap stocks. For example, on many days, the proportion of executions of hidden orders can be as large as $35$ to $40\%$ of the trades for stocks like Apple or Google.
Moreover, no information is provided as to whether the execution is for a fully hidden order, or a the tip of an iceberg order. So, we decided to remove these executions for the purpose of this first empirical study of the self-financing condition from the order book.},
this assumption holds true in all the experiments reported in this paper. In the second part of the paper, we refrain from pre-processing the data in this way and we consider the case of a general order book. For the sake of completeness we derive the self-financing equations for a general order book shape. This generalization is needed for markets where a significant amount of trades happen outside the bid-ask spread. 
As expected, this part of the paper is more involved mathematically.

We first derive discrete versions of our self financing equation and of the price impact constraint  from  NASDAQ limit order book data. Our empirical studies are done in the trade clock, and we demonstrate the significance of our microscopic analysis by rigorous statistical tests.
Next we take the limit as the tick size goes to zero, and obtain diffusion limits for both price and trade volumes. This leads to our proposed macroscopic continuous-time self-financing condition.

We propose several applications of these macroscopic equations. We first revisit local volatility models for European options in our framework and obtain hedging strategies via limit or market orders. As a highlight, we show that limit orders can only hedge negative convexity options while market orders can hedge positive convexity options. This is a rare example where the theory naturally distinguishes between the roles of liquidity providers and liquidity takers. Then a model for high frequency market making is presented to uncover the relationship between optimal spread setting and price volatility. Finally, we propose two forms of toxicity of market order flow in our continuous time setting, and for the sake of illustration, we compute their empirical analogues on the pool of 120 stocks used in a recent ECB study of HFT.
Following our theoretical analysis of general order book shapes, we propose for illustrative purposes, a supply and demand model based on perfect fill rates and deterministic price recovery.

\section{\textbf{The self-financing equation}}
\label{se:sfe}
In quantitative finance, the standard self financing portfolio equation is a cornerstone of the theory of frictionless markets. It plays a crucial role in many fundamental results, e.g. Merton's portfolio theory. Mathematically, speaking it is a simple equation which \emph{constrains} the wealth process of an investor to live in a certain sub-space. This sub-space is therefore often called the space of \emph{admissible} portfolios. 
New-comers to the mathematical theories of financial market often gripe with the self-financing condition and how it relates to the real world. While it can be postulated as a mathematical definition, it can also be \emph{derived} from a limiting procedure starting from accurate descriptions of the microstructure of trades in the trade clock. This approach is at the core of our strategy.

\begin{quote}
``The sad fact is that the self-financing condition is considerably more subtle in continuous time than it is in discrete time.''\footnote{J. Michael Steele, \emph{Stochastic Calculus and Financial Applications}, section 14.5 'Self-financing and self-doubt'.}
\end{quote}

When discussing market models at the macroscopic level, we assume that the mid-price $p$ and the inventory $L$ are given by It\^o processes:
\begin{equation}
	\begin{cases}
		dp_t &= \mu_t dt + \sigma_t dW_t \\
		dL_t &= b_t dt + l_t dW'_t
	\end{cases}
\end{equation}
for two Wiener processes  $W$ and $W'$ with unspecified correlation structure. We shall also consider an adapted process  $s_t$ representing (in the continuous time limit) the bid-ask spread measured \emph{in tick size}.
The standard self-financing condition of continuous time finance can be stated as a constraint:
\begin{equation}
\label{fo:csf}
dX_t = L_t dp_t
\end{equation}
between the price $p$ of the underlying interest, the inventory $L$, and the wealth $X$ of the agent.
In most classical financial applications, case in point Merton's portfolio theory, the price $p$ is exogenously given,  the inventory $L$ is the agent's input, and his wealth $X$ appears as the output of equation \eqref{fo:csf}. 

The objective of this paper is to generalize the self-financing portfolio condition \eqref{fo:csf} to incorporate known idiosyncrasies of the high frequency markets including transaction costs, price impact and price recovery. Also, we want this generalization to be able to quantify the differences between trading via limit orders and market orders.
We warn the reader that the equations proposed in this paper are only \emph{necessary} conditions and that   quantifying limit order fill rates, priorities and  price recovery are beyond the immediate scope of the present paper.

\subsection{Our basic formula}
The empirical analysis of NASDAQ order book data given in Section \ref{se:empirical} and in the Appendix, together with the diffusion limit arguments of Section \ref{sec:bid-ask}, prompt us to formulate the self-financing condition in the following form:
\begin{equation}
\label{eq:Intro_Formula}
dX_t = L_t dp_t \pm \frac{s_t l_t}{\sqrt{2 \pi}} dt +d[L,p]_t 
\end{equation}
where $\pm$ is $+$ when trading with limit orders, and $-$ when trading with market orders. 
Indeed, we show in Section \ref{se:empirical} below that, when time is measured in the trade clock, the discrete time analog of formula \eqref{eq:Intro_Formula} can be derived rigorously from a specific limit order book feature, and matches real wealth data extremely accurately.
We shall also impose the constraint 
\begin{equation}
\label{fo:corr_constraint}
d[L,p] < 0
\end{equation}
whenever trading with limit orders. Again, this \emph{adverse selection} constraint is also dictated by the empirical analysis of the NASDAQ data. 
\vskip 2pt
We now explain how our condition \eqref{eq:Intro_Formula} and the adverse selection constraint \eqref{fo:corr_constraint} relate to the conditions used in the separate sets of works reviewed in the introduction.

\subsection{The Almgren-Chriss model}
The seminal work by Almgren and Chriss \cite{Almgren} addresses a question  closely related to ours. These authors propose a \emph{ macroscopic model} for the price impact and the change of wealth after a liquidity taker's decision. The model leads to a very tractable framework which was used by many optimal execution studies (see \cite{Alfonsi, Wang} for example). This framework can be summarized by the system:
\begin{equation}
	\begin{cases}
		dp_t &= f_t(l_t) dt + \sigma_t dW_t \\
		dL_t &= l_t dt \\
		dX_t &= L_t dp_t - c_t(l_t)dt
	\end{cases}
\end{equation}
where $f$ and $c$ are two function-valued adapted processes which are positive, and in the case of $c$, convex.

The main advantage of this model is that price impact appears in a tractable fashion. Indeed, it comes through the function $f_t$, which creates a positive `correlation' between traded volumes and the price process. However, it constrains $L$ to be \emph{differentiable} and for this reason, the model parameters cannot be calibrated to market data directly, making the model difficult to test empirically. 
As the empirical analysis of NASDAQ data reported in Section \ref{se:empirical} and the appendix shows, there is ample evidence supporting 
nondifferentiable inventories.
Moreover, limit orders are not part of the discussion in the Almgren-Chriss framework.

\subsection{Transaction cost literature}
The branch of classical mathematical finance most related to our paper is portfolio selection under transaction costs (\cite{ Chellathurai, Magill, Shreve} or the recent review \cite{MuhleKarbe}). Most of these works start from  \emph{a model for the wealth of a liquidity taker} which generalizes the self-financing equation to a setting with transaction costs. In general however, these papers do not emphasize the derivation of the model, but instead, the study of its consequences. We hope to appeal to this side of the community by providing more accurate equations for self-financing portfolios while keeping similar tractability, leading the way to problems related to \emph{liquidity provision}, such as market making. An interesting feature of such problems is that the agent does \emph{not} directly control his portfolio, adding an additional modeling challenge. For the record we note that the standard equation used in this branch of the literature is
\begin{equation}
\label{eq:Intro_standard}
dX_t = L_t dp_t - \frac{s_t}{2} |dL|_t 
\end{equation}
where again, the inventory process $L$ is assumed to have finite variation $\int_0^t |dL|_s <\infty$ for all finite $t$ and $s_t$ is the bid-ask spread.

Strengths of this model are its simplicity, relative tractability, and straightforward calibration to the market.
Its weaknesses include the fact that the process $L$ can only have finite variation. Moreover, price impact, limit orders and other microstructure considerations are absent in the model.

Formula \eqref{eq:Intro_standard} is much closer to our proposed equation \eqref{eq:Intro_Formula} than it may seem at first. It merely corresponds to a different diffusion limit. It can be recovered in our framework by considering \emph{non-vanishing} bid-ask spread, \emph{zero} price impact and looking at market orders only. Notice that these assumptions may be more natural than ours for low frequency markets. This is presumably the reason for their introduction.

\section{\textbf{Empirical study and discrete equations}}
\label{se:empirical}

We first recall standard terminology from the high frequency markets.

\subsection{High frequency terminology}

Trading on high frequency markets takes place on an object called \emph{the limit order book}. An agent can interact with others via two possible trading mechanisms: limit orders and market orders. Limit orders correspond to the act of \emph{providing liquidity} to the market, while market orders \emph{take liquidity} from it. We will refer to agents who engage in the first type of trade as \emph{liquidity providers}\footnote{Of which \emph{market makers} are a special class.} while traders who trade with market orders will be referred to as \emph{liquidity takers}. In real markets, traders often switch between liquidity providing and taking strategies, blurring this definition somewhat.
The following comments can help highlight the differences.
\begin{itemize}
\item A liquidity taker pays a fee for his aggressiveness. This fee typically takes the form of the bid-spread, which is where most trades happen. The corresponding provider captures this bid-ask spread.
\item Right after the trade happens, the price may move. If it does, it almost always moves \emph{in favor} of the market order, compensating to some degree the transaction costs. This phenomena is called \emph{price impact}. It is a consequence of the \emph{adverse selection} of limit orders by takers.
\item Between two successive trades, the price reverts to some value in between the impacted price and the original one. \emph{Price recovery} is an intuitive name often used to describe this high frequency feature.
\item Takers control their inventory directly. Attaining \emph{correlation} with the market requires high frequency \emph{predictions} of the next price move.
\item Providers do not directly control their inventory, but only their exposure to the flow of market orders. How much of the flow they are able to capture depends on their limit order \emph{fill rate}. Flow is considered \emph{toxic} if it leads to adverse selection. The profitability of a provider's strategy depends on the spread he captures and the toxicity of his flow.
\end{itemize}

\subsection{Data used in the Study}
The statistical tests reported in this paper were produced by the analysis of the NASDAQ ITCH data of, amongst other stocks, the pool of 120 stocks used in the recent ECB study \cite{ECB} of high frequency trading.
The figures included in this paper were produced using the data for Coca Cola (KO) on 18/04/13 . As explained in an earlier footnote, the only cleaning pre-processing of the raw data was to remove the special deals and the executions of hidden orders. 

The data do not contain the identity of the agents involved in the transactions. For that reason, all quantities relating to the inventory $L$, cash $K$ or wealth $X$ are aggregate quantities which could be thought of as relating to a \emph{representative aggregate liquidity provider}. The mid-price will be denoted by $p$ and the bid-ask spread by $s$. The time stamps of the transactions are  measured in fractions of microseconds and given in the \emph{calendar} clock. However, the data analysis is performed in the \emph{trade clock} $n=1,...N$ where each time step corresponds to one trade time. For example, $p_n = p_{t_n} = p_{t_n-}$ where $t_n$ is the $n$-th trading time in the calendar clock gives the mid-price just before the $n$-th transaction.
Limit order data happening between two trade times is the source of the changes in the best bid and best ask, (and consequently of the mid-price) and is discarded for the purpose of our analysis. More generally, if $Y$ is a discrete process, we denote by $\Delta_n Y$ the forward-looking increment $\Delta_n Y = Y_{n+1} - Y_n$.

\subsubsection{More Notation}
We denote by $s_n$ the bid-ask \textbf{spread} just before the $n$-th trade. 	In other words, $s_n$ is the difference between the best ask and the best bid, just before the $n$-th trade. We shall argue later on that the spread is of the same order of magnitude as the change in price, namely that $s_n \approx |\Delta_n p|$.
\begin{figure}[htbp]
\centerline{
\includegraphics[width=5cm]{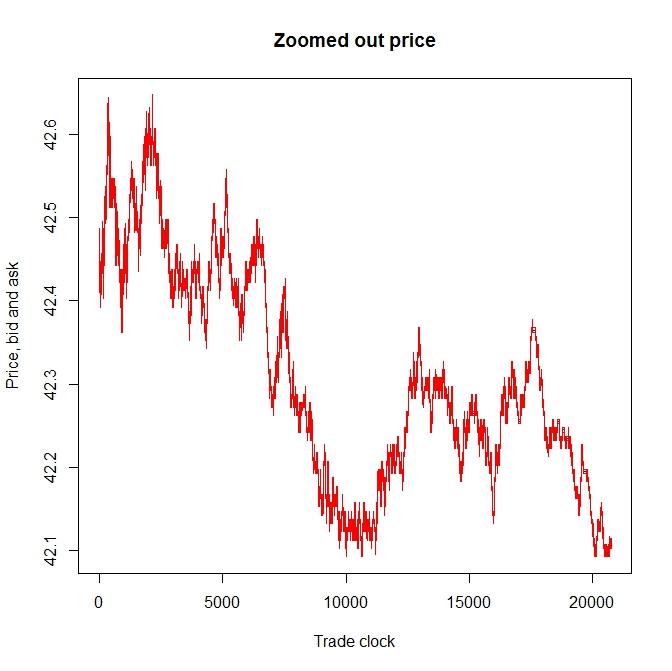}
\includegraphics[width=5cm]{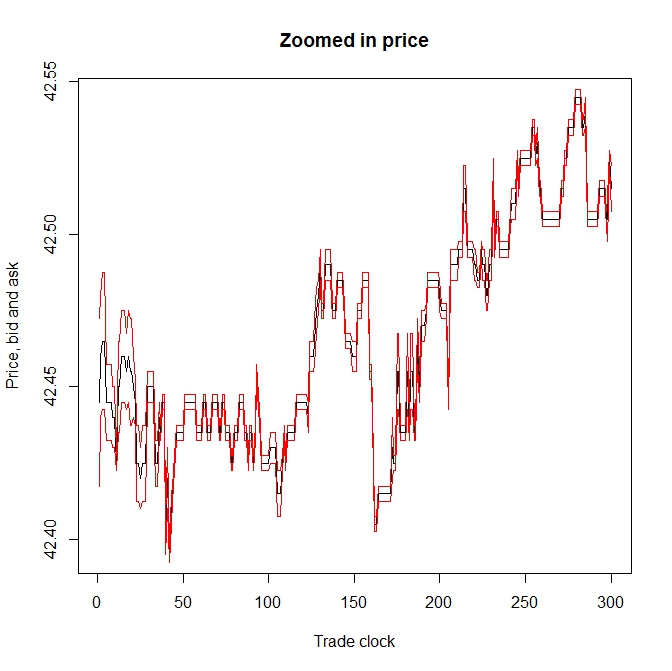}
}
\caption{Plots of the best bid, best ask and mid-price as functions of trade time (left). Zoom into a part of the graph to see the differences between the three plots (right).}
\label{fi:zoom}
 \end{figure}
We also denote by $L_n$ and $K_n$ the inventory and the cash held by the aggregate liquidity provider just before the $n$-th trade. These quantities are not given explicitly with the data provided by NASDAQ, but starting from $L_0=K_0=0$, they can easily be computed after each trade. Indeed, $L_n$ is the cumulative sum up to time $n$ of the algebraic volumes of the trades (positive volume for a limit order executed against a sell market order, and negative volume for an execution against a buy market order). Similarly, $K_n$ is the cumulative sum up to time $n$ of the cash exchanged during the trades.
The inventory and the cash $L_n$ and $K_n$ held by the aggregate liquidity provider are plotted in Figure \ref{fig:Empirical_Summary} against the trade time $n$.
\begin{figure}[htbp]
	\centering
		\includegraphics[width=1.00\textwidth]{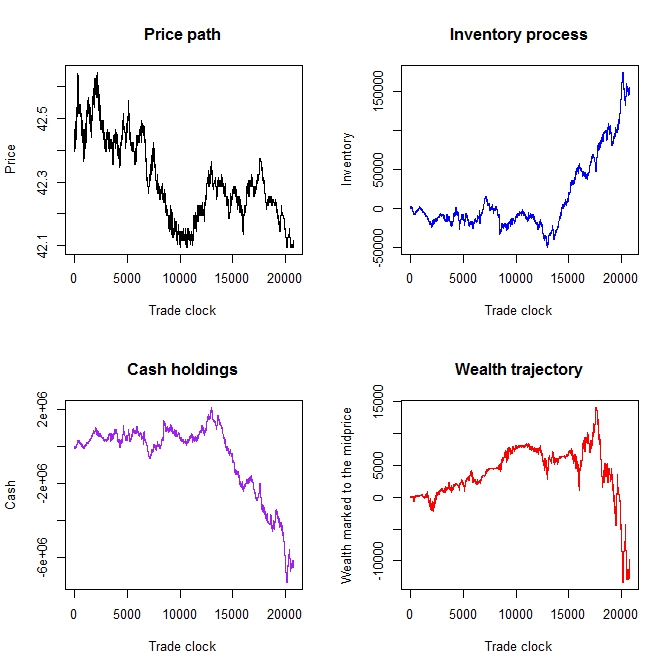}
	\caption{Coca Cola (KO) stock on 18/04/13. Inventory, cash and wealth are those of the aggregate liquidity provider.}
	\label{fig:Empirical_Summary}
\end{figure}

\subsection{Price impact}

Empirically, price impact is the simple fact that the price moves after each trade, and tends to move in favor of the market order. There have been several empirical studies and multiple proposed measures and models for it (\cite{Alfonsi, Almgren, Bouchaud, Bouchaud3, Chakraborti, Maslov2, Wang, Weber, Bouchaud2}). The main economic interpretation for price impact is adverse selection.
In this study, we isolate, measure and model price impact by a straightforward relationship.
\begin{equation}
\label{eq:Empirical_Price_Impact}
\Delta_n L \Delta_n p  \le 0. 
\end{equation} 
For all $n = 1, ... (N-1)$. This relationship states that the price cannot move up when the liquidity provider has bought and cannot move down when the provider has sold. From the taker's perspective, this means that the price always moves in his favor right after a trade.

We provide rigorous statistical tests of \eqref{eq:Empirical_Price_Impact} in Appendix \ref{se:tests}. For the sake of illustration, we note that for Coca Cola on April 18, 2013, \eqref{eq:Empirical_Price_Impact} holds for all but $166$ of the $20742$ trades of our streamlined data set, which represents $0.9\%$ of the trades. This trade impact relationship has clear consequences for the continuous time analogs of the discrete model considered here: \emph{the quadratic covariation} between the provider's inventory and the price process is \emph{negative} and \emph{decreasing}. Conversely, the quadratic covariation between the inventory of a liquidity taker and the price process is positive and increasing.

\begin{figure}[htbp]
	\centering
		\includegraphics[width=0.8\textwidth]{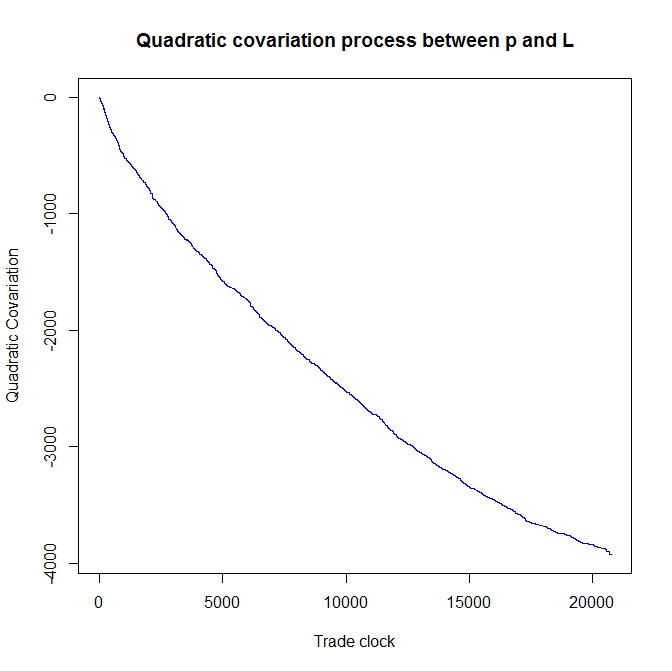}
	\caption{Quadratic covariation between inventory and price path.}
	\label{fig:quad_covar}
\end{figure}
Price impact will give us an extra compelling reason to accept trade volumes with infinite variation. Indeed, when using continuous time models, if the price path and the inventory have a non-negligible quadratic covariation, then we cannot model one as a diffusion process and the other as a finite variation process.

\begin{remark}
The causality of price impact is unclear: do trades cause price movements, or simply predict them? While not crucial for the mathematical theory, it is important for interpretation purposes, and we choose to use the second option. In particular, we shall say that a liquidity taker whose changes in inventory are \emph{strongly correlated} with the price movements has a very good \emph{short term prediction of the price}. This typically is the case of sophisticated high-frequency traders. Low-frequency traders, on the other hand, trade more slowly and acquire inventories which aren't directly correlated with the smaller price movements.
\end{remark}

\subsection{Price recovery}
This is another simple observation. Trades move prices, but typically move them \emph{at most} by one bid-ask spread. If they systematically moved the price by one bid-ask spread, then the correlation between the price path and the taker inventory should be one. Otherwise, it is smaller than one and we say that the \emph{price has recovered} from the price impact. Note that, of all our relationships, this is statistically the weakest one: it is not verified for $5\%$ of the Coca Cola data.

\begin{figure}[htbp]
	\centering
		\includegraphics[width=0.7\textwidth]{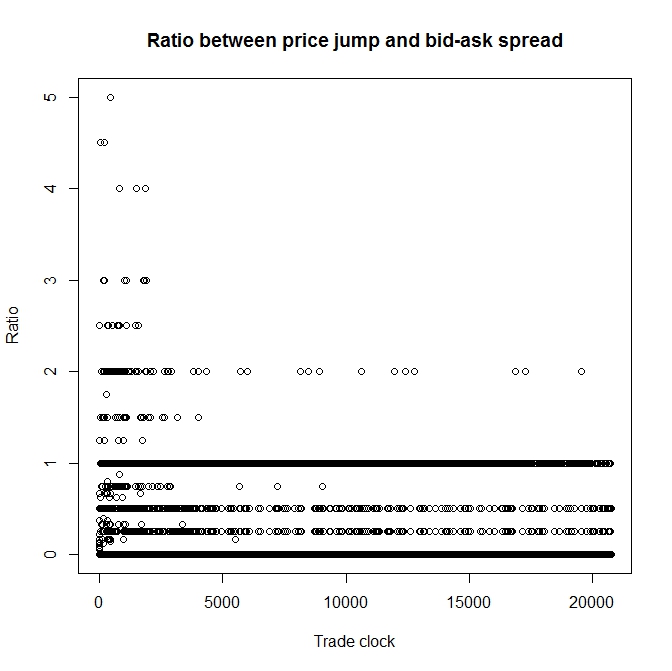}
	\caption{Relationship between price increments and spread.}
	\label{fig:ratio}
\end{figure}
Mathematically, this implies that $|\Delta_n p| \le s_n$ for $n=1, ...(N-1)$. In the continuous time version considered later, it will provide in the diffusion limit an upper bound on instantaneous price volatility based on the current spread. 

\subsection{A bit of accounting}
Finally, we derive the self-financing portfolio equation from first principles in such a high frequency market.

Because after removing the special deals and the executions against hidden orders, all the trades do happen at the best bid or ask, the amount of cash exchanged is equal to
\begin{equation}
\Delta_n K =
\begin{cases}
 - (p - \frac{s_n}{2}) \Delta_n L \text{\quad if } \Delta_n L\ge 0 \\
 - (p + \frac{s_n}{2}) \Delta_n L \text{\quad else } 
\end{cases}
\end{equation}
That is, the provider pays the bid (resp. receives the ask) when he buys (resp. sells). This can be summarized by the equation:
\begin{equation}
\Delta_n K  = - p_n \Delta_n L + \frac{s_n}{2}|\Delta_n L|\label{fo:cash}
\end{equation}

\subsubsection{The aggregate liquidity provider's wealth}
We define wealth as
\begin{equation}
\label{fo:wealth}
X_n = p_n L_n + K_n
\end{equation}
that is, the cash held by the liquidity provider plus the value of her inventory marked to the mid-price.
The wealth $X_n$ of the aggregate liquidity provider is plotted in Figure \ref{fig:Empirical_Summary} against the trade time $n$.

\subsubsection{The discrete self-financing equation}
We derive the dynamics of the wealth process $X$ from equations (\ref{fo:cash}) and  (\ref{fo:wealth}):
\begin{align}
\Delta_n X &= L_n \Delta_n p + p_n \Delta_n L +\Delta_n p \Delta_n L + \Delta_n K \nonumber\\
				 &= L_n \Delta_n p + \frac{s_n}{2} |\Delta_n L| + \Delta_n p \Delta_n L \label{eq:Empirical_Proposed}
\end{align}

\subsubsection{Empirical validation}

We compare four quantities: 1) the actual wealth, 2) the wealth computed from the standard self-financing equation:
\begin{equation}
\label{fo:ssf}
\Delta_n X = L_n \Delta_n p
\end{equation}
used in the classical Black-Scholes option pricing and Merton portfolio theories, 3) the wealth computed from the standard self-financing condition: 
\begin{equation}
\label{eq:Empirical_Standard}
\Delta_n X = L_n \Delta_n p + \frac{s_n}{2} |\Delta_n L| 
\end{equation}
advocated to include transaction costs in Merton's theory of optimal portfolio choice, and finally 4) the wealth computed from our self-financing condition \eqref{eq:Empirical_Proposed}.
\begin{figure}[htbp]
	\centering
		\includegraphics[width=0.9\textwidth]{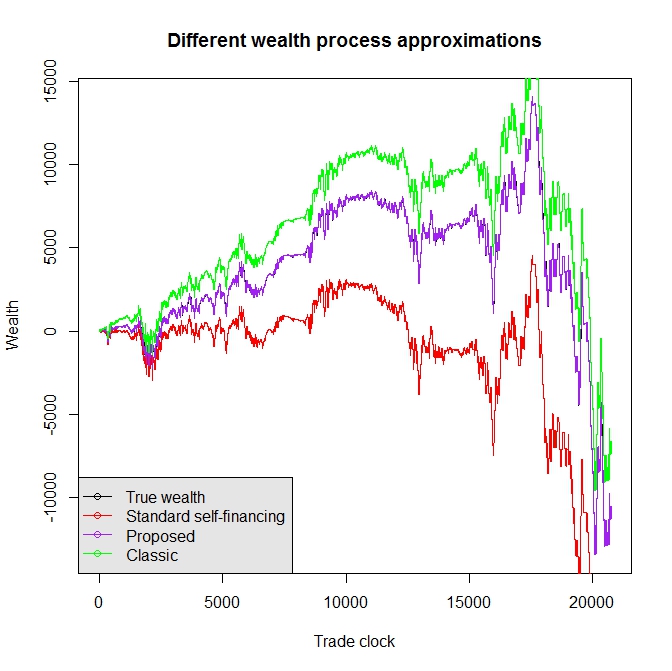}
	\caption{Plots of the actual wealth of the aggregate liquidity provider (as in Figure \ref{fig:Empirical_Summary}) together with the wealth computed from the three self-financing conditions. Red is the frictionless case. Green corresponds to  \eqref{eq:Empirical_Standard}. 
	The actual wealth and the wealth computed from our self-financing condition (\ref{eq:Empirical_Proposed}) are indistinguishable on the graph.}
	\label{fig:wealth_approx}
\end{figure}
The plots of these four wealth processes are given in Figure \ref{fig:wealth_approx} for Coca Cola stock on April 18, 2013. Changing stock or changing day does not seem to affect the following facts which are easily illustrated in this figure.
The wealth computed from the standard self-financing equation of the Black-Scholes theory clearly underestimates the actual wealth of the aggregate liquidity provider
The wealth computed from the classic equation \eqref{eq:Empirical_Standard} tries to correct for the lack of transaction cost, but it over-shoots and over-estimates the wealth of the aggregate liquidity provider. The error is reduced and practically canceled by including the adverse selection term given by the quadratic covariation, and using our proposed formula \eqref{eq:Empirical_Proposed}. The quadratic covariation between inventory and price matters!

\subsubsection{Recovering the frictionless case}

A surprising property worth mentioning concerns the case $s_n = 0$. Indeed, the latter does \emph{not} correspond to the frictionless case. Rather, choosing price jumps $|\Delta_n p| = s_n/2$ and using the fact that the price impact is negative, i.e. $\Delta_n L \Delta_n =-|\Delta_n L \Delta_n |$, yields the identity
$$
\Delta_n X = L_n \Delta_n p + \frac{s_n}{2} |\Delta_n L| - |\Delta_n p| |\Delta_n L|= L_n \Delta_n p  
$$
which is the standard self-financing portfolio equation. In our high frequency framework, it is not the absence of transaction costs that corresponds to the frictionless case, but rather the absence of price-recovery, for in that case, the price impact exactly compensates the transaction costs.

\subsection{Summary}\label{sec:Empirical_Summary}
Our empirical evidence suggests the following equations and features for the inventory $L$ and wealth $X$ of a liquidity provider, the bid-ask spread $s$ and the price $p$:

\subsubsection{Self-financing equation}
\begin{equation}
\Delta_n X = L_n \Delta_n p + \frac{s_n}{2} |\Delta_n L| + \Delta_n p \Delta_n L \label{eq:Empirical_Formula}
\end{equation}

\subsubsection{Price impact (adverse selection)}
\begin{equation}
\Delta_n L \Delta_n p \le 0
\end{equation}

\subsubsection{Price recovery}
\begin{equation}
|\Delta_n p| \le s_n
\end{equation}

\subsubsection{Vanishing bid-ask spread}
$s_n$ and $\Delta_n p$ are of the same order of magnitude, namely $s_n \approx |\Delta_n p|$.

\section{\textbf{Continuous equation: Bid-Ask case}}
\label{sec:bid-ask}
The aim of this section is to derive formula \eqref{eq:Intro_Formula} from its discrete version \eqref{eq:Empirical_Formula} established in the previous section. In the process, we shall also derive continuous-time analogs
 of the price impact / adverse selection constraint, the price recovery and vanishing bid-ask spread condition
 equivalents of the relationships of subsection \ref{sec:Empirical_Summary}. The key is to let the tick size vanish, assume that the bid-ask spread vanish with the tick size, and assume that the price and inventory converge to diffusion limits. 

\begin{remark}The bid-ask spread is of the same order of magnitude as the price jumps: the tick size. This implies in particular that the bid-ask spread vanishes in absolute terms in the diffusion limit and should therefore be measured in tick-size.
The mathematical consequence of this simple comment is that transaction costs do not diverge as the tick size goes to zero, allowing inventories that have infinite variations in the continuous limit.
\end{remark}

\vskip 4pt
The main technical tool we use in this section and section \ref{sec:General_OB} is the functional law of large number for discretized process by Jacod and Protter \cite{Jacod}. Let $\phi_{\sigma^2}$ denote the density function of the Gaussian distribution with mean $0$ and variance $\sigma^2$.

\begin{theorem}[(7.2.2) from \cite{Jacod}]
\label{thm_Jacod}
Let $(t, y) \rightarrow F_t(y)$ be an $\F_t$-adapted random function that is a.s. continuous in $(t,y)$ and verifies the growth condition $F_t(y) \le C y^2$ for some constant $C$.
Then we have the following convergence u.c.p. as $N\rightarrow \infty$ for any continuous It\^o process $Y$:
\begin{equation*}
\frac{1}{N}\sum_{n=1}^{\lfloor Nt \rfloor} F_{n/N}\left({\sqrt{N}(Y_{(n+1)/N} - Y_{n/N})}\right) \rightarrow \int_0^t \int F_s(y) \phi_{\sigma^2_s}(y)dy \; ds
\end{equation*}
where $\sigma^2_t = \frac{d[Y,Y]_t}{dt}$. 
\end{theorem}
We proceed as follows:
\begin{enumerate}
\item We begin with the continuous processes for the inventory $L$, price $p$ and bid-ask spread $s$ as our \emph{data}. 
\item By discretizing them, we obtain the \emph{data} to plug into the discrete relationships listed in subsection \ref{sec:Empirical_Summary}, yielding our \emph{discrete time} output relationships.
\item Finally, we take the limit again to obtain the diffusion limits of our discrete output to obtain our continuous-time relationships.
\end{enumerate}

In discrete time, prices are a pure-jump process, and therefore have finite variations. It is common on larger time scales to consider the price as `zoomed out' enough to be approximated by a diffusion process. Mathematically, this corresponds to a vanishing tick size. Recall that tick size is typically of the order of magnitude of the cent\footnote{Decibasis point for some exchanges in the foreign exchange market.}, that is $10^{-4}$ relative to the typical stock price. Given the relative roughness of the path of inventories when compared to prices, see for example Figure \ref{fig:Empirical_Summary}, it seems reasonable to also expect high-frequency inventories to be modeled by processes with infinite variation.

\subsection{Mathematical Setup}
Let $\left(\Omega, \F, \FF, \P\right)$ be a filtered probability space supporting two Wiener processes $W$ and $W'$ with unspecified correlation structure. We consider two It\^o processes for the price $p$ and provider inventory $L$:
\begin{equation}
	\begin{cases}
		p_t &= p_0 + \int_0^t \mu_u du + \int_0^t \sigma_u dW_u \\
		L_t &= L_0 + \int_0^t b_u du + \int_0^t l_u dW'_u
	\end{cases}
\end{equation}
where $p_0$ and $L_0$ are $\F_0$-measurable square integrable random variables, and $\mu$, $\sigma$, $b$ and $l$ are $\FF$-adapted continuous processes. Finally, we also assume the existence of a  $\FF$-adapted continuous process $s$.

Now consider the discrete approximation $p^N_n = p_{n/N}$ and likewise for $L$, $\mu$, $\sigma$, $b$ and $l$. The interpretation is that $\frac{1}{\sqrt{N}}$ is the tick size, which we formally make vanish. For the bid-ask spread $s$, we define $s^N_n = \frac{1}{\sqrt{N}} s_{n/N}$ in line with our previous comments. Plugging these definitions into the equations from subsection \ref{sec:Empirical_Summary}, we obtain the discrete relationships:

\begin{equation}
\label{fo:3discrete}
	\begin{cases}
		\Delta_n X^N = L^N_n \Delta_n p^N + \frac{s_{n/N}}{2} \frac{1}{\sqrt{N}}|\Delta_n L^N| + \Delta_n p^N \Delta_n L^N \\
		\Delta_n L^N \Delta_n p^N \le 0 \\
	  |\Delta_n p^N| \le s^N_n
	\end{cases}
\end{equation}
where the first equation is understood as the definition of the wealth $X^N$.

\subsection{Main result}
\begin{theorem}
Assuming that relations \eqref{fo:3discrete} hold for every $N\ge 1$, then the limit $\lim_{N\rightarrow \infty} X^N_{\lfloor Nt \rfloor}$ exists for the uniform convergence in probability and defines a process $X_t$ which together with the It\^o processes $p_t$ and $L_t$ satisfy the relationships:
\begin{equation}
\label{fo:3continuous}
	\begin{cases}
		dX_t = L_t dp_t + \frac{s_t l_t}{\sqrt{2 \pi}} dt +d[L,p]_t \\
		d[L,p]_t \le 0 \\
		\sigma_t \le \sqrt{\frac{2}{\pi}} s_t
	\end{cases}
\end{equation}
\end{theorem}
\begin{proof}
Using  a localizing sequence of stopping times if needed, we can assume without any loss of generality that the process $s_t$ is bounded by a constant.
The convergence of the discrete approximations of $\int_0^t L_u dp_u$ and $[L,p]_t$ is plain, proving the second relationship.

For the last term of the self-financing equation, we have that
\begin{equation}
\frac{s^N_n}{2} \frac{1}{\sqrt{N}}|\Delta_n L^N| = \frac{1}{2 N} s_{n/N} |\sqrt{N} \Delta_n L^N|
\end{equation}
which allows us to apply Theorem \ref{thm_Jacod} with $F_t(y) = \frac{s_t}{2}|y|$ and $Y_t = L_t$. This proves the self-financing equation.

The last relationship of \eqref{fo:3continuous} follows from applying the same theorem to the process $Y_t=p_t$ and the random function $F_t(y)=y^2-s_t|y|$. We obtain that, for each $t_1<t_2$, the quantity
\begin{equation}
\label{fo:t1t2}
\frac{1}{N}\sum_{n=\lfloor t_1 N\rfloor}^{\lfloor t_2 N\rfloor} \left((\sqrt{N} \Delta_n p^N)^2 - s_{n/N}|\sqrt{N}\Delta_n p^N|\right)
\end{equation} 
converges toward 
\begin{equation}
\int_{t_1}^{t_2} (\sigma_t  - \sqrt{\frac{2}{\pi}} s_t) \sigma_t dt,
\end{equation}
and the fact that this process is negative for all $t_1<t_2$ concludes the proof.
\end{proof}

\begin{remark}
Technically speaking, nothing prevents us from going through with the same limiting argument for the hidden part of the order book, simply replacing $p_t$ and $s_t$ by their `hidden' counterparts. Two practical problems appear however. First, measuring the hidden price and spread is difficult. Second, and more importantly, it is unclear by what to replace the price impact inequality, as adverse selection of hidden orders is not well studied or understood.
\end{remark}

\subsection{Time change}

Note that equation \eqref{eq:Intro_Formula} was proved in a \emph{trade clock}, which means that all the time-related quantities, such as volatility, must be measured per trade time. While this is a positive feature for high frequency models under this clock (e.g. \cite{Ane, Bouchaud}), it is less advantageous for financial problems working under a different clock. For example, pricing an option with a fixed maturity in the calendar clock may be difficult to do directly from equation \eqref{eq:Intro_Formula}. We therefore discuss how our proposed formula behaves under time-changes, with the canonical time-change being the switch to a calendar clock. Another possible time-change is the switch from a trade clock to a volume clock.

\begin{definition}
We define a good time change to be an $\F_t$-adapted stochastic process $\tau_t$ such that $\tau_0 = 0$ and 
\begin{equation}
d\tau_t = n^2_t dt
\end{equation}
with $n_t$ uniformly bounded away from zero.
\end{definition}

We start from:
\begin{equation}
\begin{cases}
dp_t &= \mu_t dt + \sigma_t dW_t \\
dL_t &= b_t dt + l_t dW'_t \\
dX_t &= L_t dp_t + \frac{s_t}{\sqrt{2\pi}} l_t dt + d[L, p]_t
\end{cases}
\end{equation}
with $d[L,p]_t \leq 0$, and we study the processes $\tilde{p}_t = p_{\tau_t}$, $ \tilde{L}_t = L_{\tau_t}$ and $ \tilde{X}_t = X_{\tau_t}$. Note that all the time-changed processes are now adapted with respect to the time-changed filtration $\tilde{\F}_t = \F_{\tau_t}$. Note also that the processes $\tilde{W}_t = \int_0^{\tau_t} 1/n_{\tau^{-1}_u} dW_u$ and $\tilde{W}'_t = \int_0^{\tau_t} 1/n_{\tau^{-1}_u} dW'_u$ are $\tilde{\F}_t$ Wiener processes.

A simple chain-rule leads to the time-changed dynamics:
\begin{equation}
\begin{cases}
d\tilde{p}_t &= \tilde{\mu}_t dt + \tilde{\sigma}_t d\tilde{W}_t \\
d\tilde{L}_t &= \tilde{b}_t dt + \tilde{l}_t d\tilde{W}'_t \\
d\tilde{X}_t &= \tilde{L}_t d\tilde{p}_t + \frac{\tilde{s}_t}{\sqrt{2\pi}} \tilde{l}_t dt + d[\tilde{L}, \tilde{p}]_t
\end{cases}
\end{equation}
where
\begin{align*}
\tilde{\mu}_t = n^2_t \mu_{\tau_t} ; \quad \tilde{b}_t = n^2_t b_{\tau_t} \\
\tilde{\sigma}_t = n_t \sigma_{\tau_t}; \quad \tilde{l}_t = n_t l_{\tau_t}
\end{align*}
which are standard, as well as the more surprising:
\begin{equation}
\tilde{s}_t = n_t s_{\tau_t}
\end{equation}
\begin{remark}
Part of this result is expected: under the modified time clock, drifts and volatility must be measured by the new unit of time instead of by unit of trade, which corresponds to the factors $n^2_t$ and $n_t$. However, the unfortunate result is that the bid-ask spread must also be multiplied by $n_t$, which means that one needs to keep track of the process $\tilde{s}_t$ rather than the more natural process $s_{\tau_t}$.
\end{remark}

\begin{remark}
This issue is resolved when $s_t = \lambda \sigma_t$. Such a assumption would follow the conclusion of the empirical paper \cite{Bouchaud2} which suggests a linear relationship between daily average bid-ask spread and daily average volatility \emph{per trade}. From a theoretical perspective, this model is stable under time change, in the sense that $\tilde{s}_t = \lambda \tilde{\sigma}_t$, a desirable property.
\end{remark}

\begin{remark}
One could have also from the beginning worked under the changed clock and used the law of large numbers with irregular discretization schemes found in \cite{Jacod} to recover the same result.  
\end{remark}

\subsection{The case of a liquidity taker}
By symmetry, the corresponding equations for the inventory and wealth of a liquidity taker are 
\begin{equation}
	\begin{cases}
		dX_t = L_t dp_t - \frac{s_t l_t}{\sqrt{2 \pi}} dt +d[L,p]_t \\
		d[L,p]_t \ge 0
	\end{cases}
\end{equation}
Unfortunately, as we already pointed out, these equations are only \emph{necessary} conditions. Indeed, unlike with the standard self-financing equation, it is difficult to tell which processes $L$ and $p$ are admissible: we can only derive $X$ once $L$ and $p$ are given.

To give an example of why not all $L$ can be attained, assume the volume on the order book is finite. Then the volatility of $L$ must be bounded by the amount of volume available. Other factors that can come into play to determine which processes $L$ are actually attainable by market participants are: limit order fill rate, instantaneous price recovery and for market orders the ability to predict the next price jump. These factors will directly impact the volatility of $L$ and the possible correlation and quadratic covariation between $L$ and $p$.

Ultimately, supply and demand rule the price $p$ and volume $L$. $X$, however, stems from accounting rules.

\section{\textbf{Applications}}
Applications of the proposed relationships depend on models of the inventory $L$ and the price $p$. Notice that, when we formulate an optimization problem, we often assume that the inventory can be any It\^o process. This is an act of faith as making it happen typically requires good execution algorithms and limit order fill rates.

Reasonable models for the spread $s$ are easier to come by. We shall typically scale the spread with the price volatility: $s_t =  \sqrt{2 \pi} \lambda \sigma_t$ (for some constant $\lambda > 1/2$). This is consistent with the empirical literature on the matter, e.g. \cite{Bouchaud2}.

\subsection{Hedging}
In this subsection we explore perfect replication of European options, assuming that the corresponding inventory can be attained via high-frequency trades. Let $f$ be the payoff function of our option and let 
\begin{equation}
dp_t = \mu(t,p_t)dt + \sigma(t,p_t)dW_t
\end{equation} 
be a Markovian stochastic differential equation for the price. Denote by $L$ the inventory of the hedger and let us assume that it is of the form:
\begin{equation}
dL_t = b_t dt + l_t dW_t
\end{equation}
with  $b_t$ and $l_t$ continuous, bounded and adapted processes. Note that the dynamics of $L_t$ are driven by the same Wiener process as the price, so the model is complete and perfect replication is possible. Note also that $l_t <0$ corresponds to trading via limit orders and $l_t >0$ to trading via market orders. Furthermore, when working with this signed $l_t$, the self-financing equation writes the same for limit and market orders:
\begin{equation}
	dX_t = L_t dp_t - \frac{s_t l_t}{\sqrt{2 \pi}} dt +d[L,p]_t 
\end{equation}
as when $l_t<0$, we want to capture the transaction costs and when $l_t>0$, we need to pay them.

Assume that interest rates are zero and define by $v(t,p)$ the price of the option knowing that $p_t = p$. Then we have the replication equation
\begin{align*}
d(X_t - v(t,p_t)) = \left(L_t - \Delta_t\right)dp_t + d[p,L]_t - \frac{s_t}{\sqrt{2\pi}} l_t dt - (\Theta_t + \frac{1}{2}\Gamma_t \sigma^2(t,p_t))dt
\end{align*}
where $\Delta_t$, $\Theta_t$ and $\Gamma_t$ denote the usual Greeks evaluated at $t$ and $p_t$. Delta hedging the option removes the price risk and leads to the equation
\begin{align*}
dL_t &= d\Delta_t \\
		 &= \left(\partial^2_{t} v(t,p_t) + \frac{1}{2} \sigma^2(t,p_t)\partial^3_{tp^2} v(t, p_t)\right)dt + \Gamma_t dp_t 
\end{align*}
and in particular the identity
\begin{equation}
l_t = \Gamma_t \sigma(t, p_t)
\end{equation}
Note that therefore $l_t$ and $\Gamma_t$ must be of the same sign!

Finally, the pricing partial differential equation becomes
\begin{equation}
\partial_t v(t, p) +  \left(\lambda -\frac{1}{2}\right) \sigma^2(t, p) \partial^2_p v(t,p)  = 0
\end{equation}
with terminal condition $v(T,p) = f(p)$. As $\lambda > 1/2$, this leads to a multiplicative factor of $\sqrt{2 \lambda - 1}$ on the implied local volatility when compared to the frictionless case.

\begin{remark}
An important point is that negative Gamma options can be replicated via limit orders, while positive Gamma options can be replicated via market orders. This is assuming that one can guarantee perfect correlation (respectively anti-correlation) with the price process for the inventory $L_t$ to be driven by the same Wiener process as the price. 

This is consistent with the fact that one would not expect to use limit orders to delta-hedge a call option, as hedging a call option requires you to buy when the price goes up and sell when the price goes down: exactly the opposite of a limit order.
\end{remark}

\subsection{Market making}

In this section, we adapt to our framework the key insight of the model proposed in \cite{Stoikov}. The ultimate aim is to solve the optimization problem of a representative market maker choosing the spread and maximizing his profits. The trade-off he faces, and which is the key ingredient of the model, is the following: the smaller the spread, the likelier trades are, but the less profit he makes on each of them.

In a way similar to \cite{Stoikov, Stoikov2}, we model the probability of execution of a limit order by a decreasing function of the quoted spread. This will first be done at the microscopic level, to obtain a reasonable model for our inventory process $L$ at the macroscopic level. A key difference with \cite{Stoikov} is that we still impose the price impact constraint, which will further depress the market maker's profits because of adverse selection.

To guarantee the price impact constraint is satisfied, we use, at the microscopic level, a modified version of the Almgren and Chriss model \cite{Almgren} to relate the price to the aggregate inventory of the liquidity providers.
We assume that
\begin{equation}
\Delta_n L = - \lambda_{n+1} \Delta_n p
\end{equation}
for a $\F_{n+1}$-measurable, positive random variable $\lambda_{n+1}$. This is an unpredictable form of linear price impact, in the sense that, ex-post, the price increment is a linear function of the traded volume.

To capture the insight of \cite{Stoikov}, we model $\lambda_{n+1}$ in such a way that
\begin{equation}
\E[\left.\lambda_{n+1}\right|\F_n] = \rho_n(s_n) f_n(s_n); \quad \E[\left.\lambda^2_{n+1}\right|\F_n] = \left(f_n(s_n)\right)^2
\end{equation}
where $s_n$ is the market maker's chosen spread, and $\rho_n$ and $f_n$ are continuous, positive function with $f_n$ decreasing and $\rho_n \in [0,1]$. The assumption that $f_n$ is decreasing in the spread is inherited from \cite{Stoikov}, and the fact that $\rho$ must be smaller than $1$ is due to Jensen's convexity inequality. We assume $\lambda_{n+1}$ to be independent of $\Delta_n p$ conditional on $\F_n$.
Computing the predictable quadratic variation of $L_n$ yields:
\begin{equation}
 \sum_{k=1}^{n-1} f^2_k(s_k) \E\left[\left.\Delta_k p^2\right|\F_k\right],  
\end{equation}
while the predictable quadratic covariation of $L_n$ and $p_n$ is given by:
\begin{equation}
- \sum_{k=1}^{n-1} \rho_k(s_k) f_k(s_k) \E\left[\left.\Delta_k p^2\right|\F_k\right].  
\end{equation}
This suggests the use of the following model in the continuum limit:
\begin{equation}
\begin{cases}
dp_t &= \mu_t dt + \sigma_t dW_t \\
dL_t &= -\rho_t(s_t) f_t(s_t) \mu_t dt + f_t(s_t)\sigma_t dW'_t 
\end{cases} 
\end{equation}
with $d[W,W']_t = -\int_0^t \rho_u(s_u) du$ for some adapted, continuous  and positive functions $\rho_t(\cdot)$ and $f_t(\cdot)$ with $\rho_t \le 1$ and $f_t$ decreasing.
Note that the equation for $L_t$ can also be rewritten as:
\begin{equation}
dL_t = -\rho_t(s_t) f_t(s_t) dp_t + f_t(s_t)\sqrt{1 - \rho^2_t(s_t)}\sigma_t dW^\perp_t
\end{equation}
with a Wiener process $W^\perp_t$ independent from $W_t$. We will from now on assume that $p_t$ is adapted to the filtration generated by $W_t$.

Applying our wealth equation, we obtain:
\begin{equation}
X_T = L_T p_T - \int_0^T p_t dL_t + \frac{1}{\sqrt{2\pi}} \int_0^T \sigma_t  s_t f_t(s_t)dt .
\end{equation}
For both $f_t$ and $\rho_t$, a natural assumption is that they are functions of the spread rescaled by the volatility:
\begin{equation}
f_t(s) = f(s/\sigma_t); \quad \rho_t(s_t) = \rho(s_t/\sigma_t)
\end{equation}
for some $C^0$ decreasing function $f$ and $C^0$ function $\rho$. We will furthermore assume that $g(x) = xf(x)$ is a decreasing function for $x$ large enough, that $g(x) \rightarrow 0$ as $x\rightarrow \infty$, and that $f(x)>0$ for all $x\ge 0$.

The problem of a risk-neutral market maker attempting to set the spread optimally is to maximize:
\begin{equation}
\sup_{s} \E X_T.
\end{equation}
We solve this control problem using the Pontryagin maximum principle. Let us define a few functions first.

\begin{lemma}
For all $a>0$, define the function $F_a$ by
\begin{equation}
F_a : x \mapsto \frac{x}{\sqrt{2\pi}} f(x) - a \rho(x) f(x)
\end{equation}
Then the function 
\begin{equation}
M(a) = \max_{x \in [0,\infty)} F_a(x)
\end{equation}
is well defined, continuous, and decreasing in $a$. Furthermore, there exist a measurable selection
\begin{equation}
m(a) \in \textit{argmax}_{x \in [0,\infty)} F_a(x)
\end{equation}
and we have that $m(a) >0$.
\end{lemma} 
\begin{proof}
First, note that for all $a>0$, 
$$
F_a(0) = - a \rho(0) f(0) \le 0, \quad F_a(a+1) \ge f(a+1) >0
$$
Next, if $g$ is decreasing on the interval $[x_0, \infty)$, then we can define the function $\beta(a)$ as
$g^{-1} \circ f(a+1)$ if $f(a+1)$ is in $g[x_0,\infty)$, and $x_0$ otherwise. $\beta(a)$ is continuous and verifies $f(a+1) \ge g(x)$ for all $x\in (\beta(a), \infty)$.

This proves that the maximum of $F_a$ is attained on the compact $[a+1, \beta(a)]$. The continuity of $M$ holds by Berge's maximum theorem. It is decreasing by definition of $F_a$. The measurable selection result follows by Thm 18.19 of \cite{Aliprantis}.
\end{proof}

\begin{proposition}
Any solution of the control problem is of the form
\begin{equation}
\frac{s_t}{\sigma_t} = m\left(\alpha_t\right)
\end{equation}
where 
\begin{equation}
\alpha_t = \E\left[\left. p_T - p_t\right|\F_t\right] \frac{\mu_t}{\sigma^2_t} + \frac{Z_t}{\sigma_t},
\end{equation}
$Z_t$ being the volatility of the martingale representation of $p_T$
\end{proposition}
\begin{proof}
We apply the necessary part of the stochastic Pontryagin maximum principle.
The generalized Hamiltonian is equal to:
\begin{eqnarray*}
&&\mathcal{H}_t(s, L, Y, Z, Z^\perp) = - \rho(s/\sigma_t) f(s/\sigma_t)\left[\left(Y_t -p_t\right)\mu_t + \sigma_t Z\right] \\
&&\phantom{???????????????????}+ \frac{\sigma_t}{\sqrt{2 \pi}} s f(s/\sigma_t) +  \sigma_t f(s/\sigma_t) \sqrt{1 - \rho^2(s/\sigma_t)} Z^{\perp}
\end{eqnarray*}
and the adjoint equation is solved by 
\begin{equation}
Y_t = \E\left[\left. p_T \right|\F_t\right]
\end{equation}
which, in particular, implies $Z^\perp_t =0$. $Z_t$ can be computed via the martingale representation theorem on the Brownian filtration generated by $W_t$.

The Hamiltonian to maximize therefore becomes
\begin{equation}
\sigma^2_t F_{\alpha_t} \left(\frac{s}{\sigma_t}\right)
\end{equation}
and the previous lemma concludes.
\end{proof}

Beyond the optimal control, one might be interested in the dependence in $\sigma_t$ and $\alpha_t$ of the market maker's expected profits as well as the volatility of his inventory. Note that a low volatility of the inventory means that the market maker has essentially pulled out of the market.

\begin{corollary}
The market maker's expected profits and losses are
\begin{equation}
\E\left[\int_0^T M\left(\alpha_t\right) \sigma^2_t dt\right]
\end{equation}
while the volatility of his inventory is 
\begin{equation}
\sigma_t f (m\left(\alpha_t\right)).
\end{equation}
\end{corollary}

\begin{proof}
The expected profits can be computed by integrating the Hamiltonian along the optimal path. The rest follows from the previous proposition.
\end{proof}

A consequence of the corollary is that the market maker is on average short $\alpha_t$ and, for $\alpha_t$ being fixed, long volatility.

There are now two distinct problems if one looks for tractable formulas. First, an explicit model for $p_T$ must be given for which the martingale representation term $Z_t$ can be computed. Second, one has to propose a function $g$ for which the maximal argument $m$ of $F$ can easily be characterized as a function of $\alpha_t$. 

\subsubsection{The martingale case}
Note that the latter problem is solved when $p_t$ is assumed to be a martingale. Indeed, if we have
\begin{equation}
dp_t = \sigma_t dW_t
\end{equation}
for some adapted, continuous and positive process $\sigma_t$. Then $\alpha_t = 1$ and we simply have 
\begin{equation}
s_t = m(1) \sigma_t
\end{equation}
circumventing the need for explicit functions $\rho$ and $f$. This result provides a theoretical argument for the empirical claim made in \cite{Bouchaud2} that the spread is a linear function of volatility.

Plugging this optimal spread back into the objective function, the market maker's expected profits and losses (P\&L) are
\begin{equation}
M(1) \E\left[\int_0^T \sigma^2_t dt\right]
\end{equation}

In the martingale case, the market maker is therefore \emph{on average}, Delta neutral, has negative Gamma but positive Vega.

\subsubsection{Explicit cases}
Other cases where $\alpha_t$ can be computed explicitly are:
\begin{itemize}
\item the Black-Scholes model
\begin{equation}
dp_t = \mu p_t dt + \sigma p_t dW_t
\end{equation}
in which case we obtain:
\begin{equation}
\E \left[\left. p_T\right|\F_t\right] = p_t e^{\mu(T-t)}; \quad Z_t = \sigma p_t e^{\mu(T-t)},
\end{equation}
and hence 
\begin{equation}
\alpha_t = \frac{\mu}{\sigma^2} \left(e^{\mu(T-t)} -1\right) + e^{\mu(T-t)}.
\end{equation}

\item the case of a mean reverting (Ornstein-Uhlenbock) price process 
\begin{equation}
dp_t = \rho (p_0 - p_t) dt + \sigma dW_t
\end{equation}
in which case: 
\begin{equation}
\E \left[\left. p_T\right|\F_t\right] = p_0 + e^{-\rho (T- t)}(p_t - p_0); \quad Z_t = \sigma e^{- \rho(T - t)},
\end{equation}
and hence
\begin{equation}
\alpha_t =   - \frac{\rho}{\sigma^2} \left(p_t - p_0\right)^2\left(e^{-\rho(T-t)} -1\right) + e^{-\rho(T-t)}.
\end{equation}
\end{itemize}

Unlike in the martingale case, it is hard to obtain any tractable formulas without specifying a functional form for $\rho$ and $f$. In the case where $\rho(x) = 1/(1+x)$ and $f(x) = 1/(1+x)^2$, the optimal spread becomes
\begin{equation}
s_t = \sigma_t \sqrt{1 + 3 \alpha_t}
\end{equation}

Note that $m$ is an increasing function of $\alpha_t$. To compare with the martingale case, where $\alpha_t = 1$, we therefore want to compare the ratio $\alpha_t$ to $1$ to study the impact of the model assumptions on the market maker's profits and inventory volatility.

\begin{itemize}
\item For the Black-Scholes model, $\alpha_t$ is larger than $1$ for $\mu >0$. For $\mu<0$, there exists a critical value depending on $T$ and $\sigma$ for which this ratio flips sign.

\item In the case of an Ornstein-Uhlenbock process, $\alpha_t$ is smaller than $1$ iff
\begin{equation}
\left(p_t - p_0\right)^2 < \frac{\sigma^2}{\rho}
\end{equation}
that is, if the current price $p_t$ isn't too far from the long-term average $p_0$.
\end{itemize}
\vskip 4pt\noindent
In line with intuition, the market maker quotes larger spreads, expects less profit, and captures less volume in the 'momentum' Black-Scholes model, as compared to the martingale case. In a mean-reverting market, unless the price is significantly away from its long-term trend, the market maker quotes smaller spreads, expects more profit and captures more volume than in the two other market models.

\subsection{Transaction cost analysis and measure of toxicity}

Following the suggestion of \cite{O'Hara}, one aim of the analysis is to provide macroscopic \emph{analysis tools} of microstructure for LFTs and academics. Not everyone wants to delve into the details of high frequency rules. In this respect, this paper only scratches the surface of the microstructure relationships HFTs can uncover, but it conveniently summarizes them and compares them to the standard 'frictionless' case.

\cite{O'Hara} identifies two particular tools that could be of use. One is what the paper calls 'transaction cost analysis', which we interpret to be the analysis of the difference between the effective wealth, and the one that would have been obtained in a frictionless market. Therefore, 'transaction costs' contain two terms:
\begin{itemize}
\item the spread component: 
\begin{equation}
 \pm \int_0^T\frac{s_t l_t}{\sqrt{2 \pi}} dt
\end{equation}
This component is positive if using limit orders, and negative if using market orders. Using one or the other affects the Gamma exposure of the trading strategy.
\item and the price impact component:
\begin{equation}
[L,p]_T
\end{equation}
which is always of the opposite sign to the spread component.
\end{itemize}
Depending on the Gamma of the LFT strategy, one or the other term will be the potential source of losses of the trader.

The second tool sought for is a measure of toxicity of the flow of market orders, preferably expressed as an index. Such an index could be used both by market makers to decide on whether it is profitable to provide liquidity and by LFTs to decide whether to execute their trades now or wait for better market conditions. The toxicity of market orders is entirely captured in our framework by the price impact term $[L,p]_T$. Two natural measures of the strength of this price impact term, and hence toxicity of market order flows, are as follows:
\begin{itemize}
\item The instantaneous negative correlation 
\begin{equation}
\label{fo:ct_toxicity}
\rho_t = - \frac{1}{\sigma_t\ell_t}\frac{d[L,p]_t}{dt}
\end{equation}
between the aggregate provider's inventory and the price. In particular, this could serve as a benchmark for a particular market maker to measure if the flow of market orders he captures is more or less toxic than that of the market as a whole. For the purpose of empirical studies, when working in the discrete trade clock, we compute the discrete time toxicity index as the negative of the empirical correlation of the inventory and the mid-price over the time interval $[0,t]$, namely:
\begin{equation}
\label{fo:dt_toxicity}
\rho^{(d)}_t = - \text{corr}(\Delta L_{[0,t]},\Delta p_{[0,t]})
\end{equation}
which is nothing but a plain discretization of formula \eqref{fo:ct_toxicity}.
\item The ratio between the integrated price impact and spread components of the aggregate provider's wealth.
\begin{equation}
r = -\sqrt{2\pi}\frac{[L,p]_T}{\int_0^T s_t l_t dt}
\end{equation}
which can be discretized as
\begin{equation}
r^{(d)} = - 2 \frac{\sum \Delta_n p \Delta_n L}{ \sum s_n |\Delta_n L|}
\end{equation}
The market maker in particular holds an implicit option on this quantity: he can pull out of the market if the ratio is larger than $1$, as in that case he loses money even in the absence of long term alpha trading by his LFT clients. 
\end{itemize} 
The advantage of the first measure of toxicity is that it measures the immediate proportion of toxic versus non-toxic market orders. The disadvantage is that it must be estimated via statistical procedures. The second measure, on the other hand, is more closely related to the actual P\&L of a market maker but must be computed over a longer time horizon, making it an ex-post analysis tool.

We give a table illustrating these two measures across several stocks on a same given trading day.
\begin{table}[htbp]
	\centering
		\begin{tabular}{l|c|c }
						Stock			&	$\rho^{(d)}$    & $r^{(d)}$ \\
									\hline
					 AAPL    &  0.17270704 & 0.19904208 \\
					 GOOG    &  0.23689058 & 0.32856196 \\
					 BRCM    &  0.19237560 & 0.29776003 \\
					 CELG    &  0.26835355 & 0.48287317 \\
					 CTSH    &  0.33887494 & 0.51758560 \\
					 CSCO    &  0.08393210 & 0.09300757 \\
					 BIIB    &  0.27832205 & 0.40193651 \\
					 AMZN    &  0.23614694 & 0.30494250 \\
					 GPS     &  0.20956508 & 0.48908889 \\
					 SFG     &  0.24173454 & 0.57253111 \\
					 INTC    &  0.05301259 & 0.05574866 \\
					 GE      &  0.10889870 & 0.11888714 \\
					 JKHY    &  0.33407745 & 0.56987813 \\
					 PFE     &  0.15849674 & 0.15958849 \\
					 CBT     &  0.34887086 & 0.74490980 \\
					 AGN     &  0.35890531 & 0.78020785 \\
					 CB      &  0.38667565 & 0.58090719 \\
					 AA      &  0.08046277 & 0.08406282 \\
					 FPO     &  0.49598056 & 1.14964119 \\
		\end{tabular}
			\caption{Values of the toxicity indexes on sample stocks.}
\end{table}

\section{\textbf{Continuous equation: general order book shape}}\label{sec:General_OB}
While on our particular choice of stock, most of the trades happened at the best bid or ask price, we wish to generalize our results to a general limit order book. This section starts by formally introducing the notion of limit order book and deriving some basic machinery before going through with the same diffusion limit strategy as section \ref{sec:bid-ask}.

\subsection{Microscopic description of the order book}

Borrowing from a time-honored method in statistical physics, we first describe in depth the interactions between agents at a microscopic level before deriving effective equations holding at the macroscopic level. We consider a \emph{single} liquidity taker and a \emph{single} liquidity provider. They trade an asset whose possible price range is $(0,\infty)$ via a limit order system. The liquidity provider always moves first by choosing the \emph{limit orders} she places on the limit order book. These limit orders are represented by a control variable $(b,a)$ consisting of a pair of strictly positive measures on $(0,\infty)$. The liquidity taker then chooses the control variable $(\beta,\alpha)\in(0,\infty)\times (0,\infty)$ representing market orders that he wants to execute on that limit order book.

Throughout this section we use the liquidity provider's point of view to track changes in portfolio positions and ignore the following high frequency phenomena:
\begin{enumerate}
\item \emph{Slippage}. Market orders execute immediately at their intended price.
\item \emph{Partial fills}. Market orders consume all the volume present at a given price\footnote{This property automatically holds when you formally consider a continuous order book distribution.}.
\item \emph{Hidden orders}. All limit orders are public.
\end{enumerate}

\subsubsection{Basic relationships}
We first focus on \emph{basic} relationships between the two agents, their orders and inventories. 

\vskip 5pt
\hskip -15pt
\begin{minipage}{.4\textwidth}
			\includegraphics[width=1\textwidth]{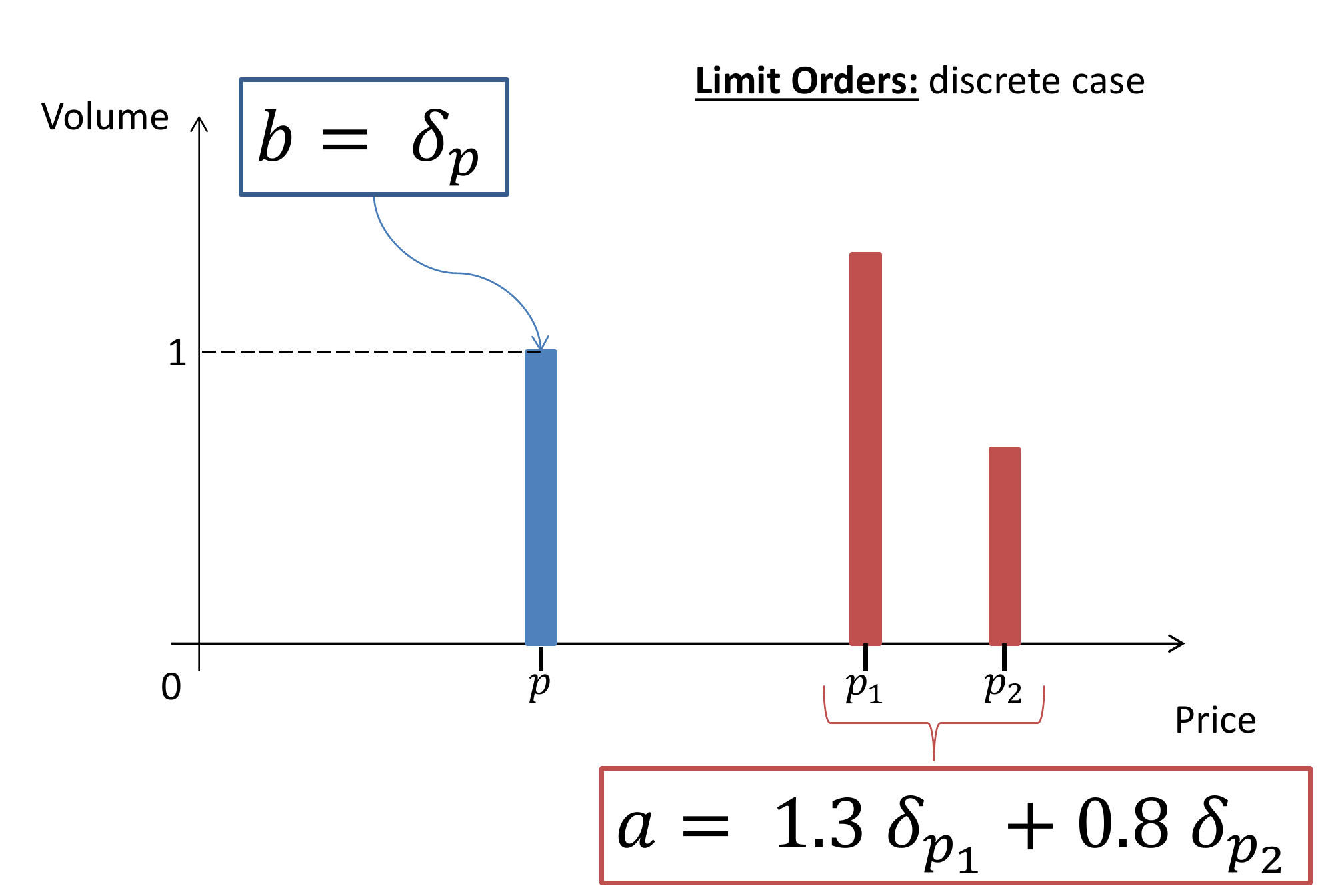}
\end{minipage}
\begin{minipage}{.6\textwidth}
The control $(b,a)$ of the liquidity provider represents her \emph{limit orders}. A bid for one unit of the asset placed at a price $p$ is represented by the probability measure $b = \delta_p$, while an offer (or ask) for one unit at price $p'$ by $a = \delta_{p'}$. If the provider places multiple limit orders, we sum these unit masses and obtain two non-negative measures $b$ and $a$ representing the liquidity provider's aggregate orders.
\end{minipage}
\vskip 5pt

We will call $(b,a)$ a \emph{limit order book}, or order book. We define the best bid and ask of an order book in the following way:

\begin{definition}[Best bid and ask]
Let $(b,a)$ be an order book. Then we define the best bid and best ask prices to be
\begin{equation}
\bar{b} = \sup\{p\in\text{supp}(b)\} \;,\quad \underline{a} = \inf\{p\in\text{supp}(a)\}
\end{equation}
Here we use the notation $\text{supp}(\mu)$ for the topological support of the measure $\mu$.
\end{definition}

\vskip 5pt
\hskip -15pt
\begin{minipage}{.6\textwidth}
\begin{remark}
In real markets, such limit orders can only be placed on a discrete grid, and the resulting $a$ and $b$ are always discrete measures. The recent push of high frequency markets to refine their grid may justify considering measures $a$ and $b$ that are absolutely continuous with respect to the Lebesgue measure.
\end{remark}
\end{minipage}
\begin{minipage}{.4\textwidth}
			\includegraphics[width=1\textwidth]{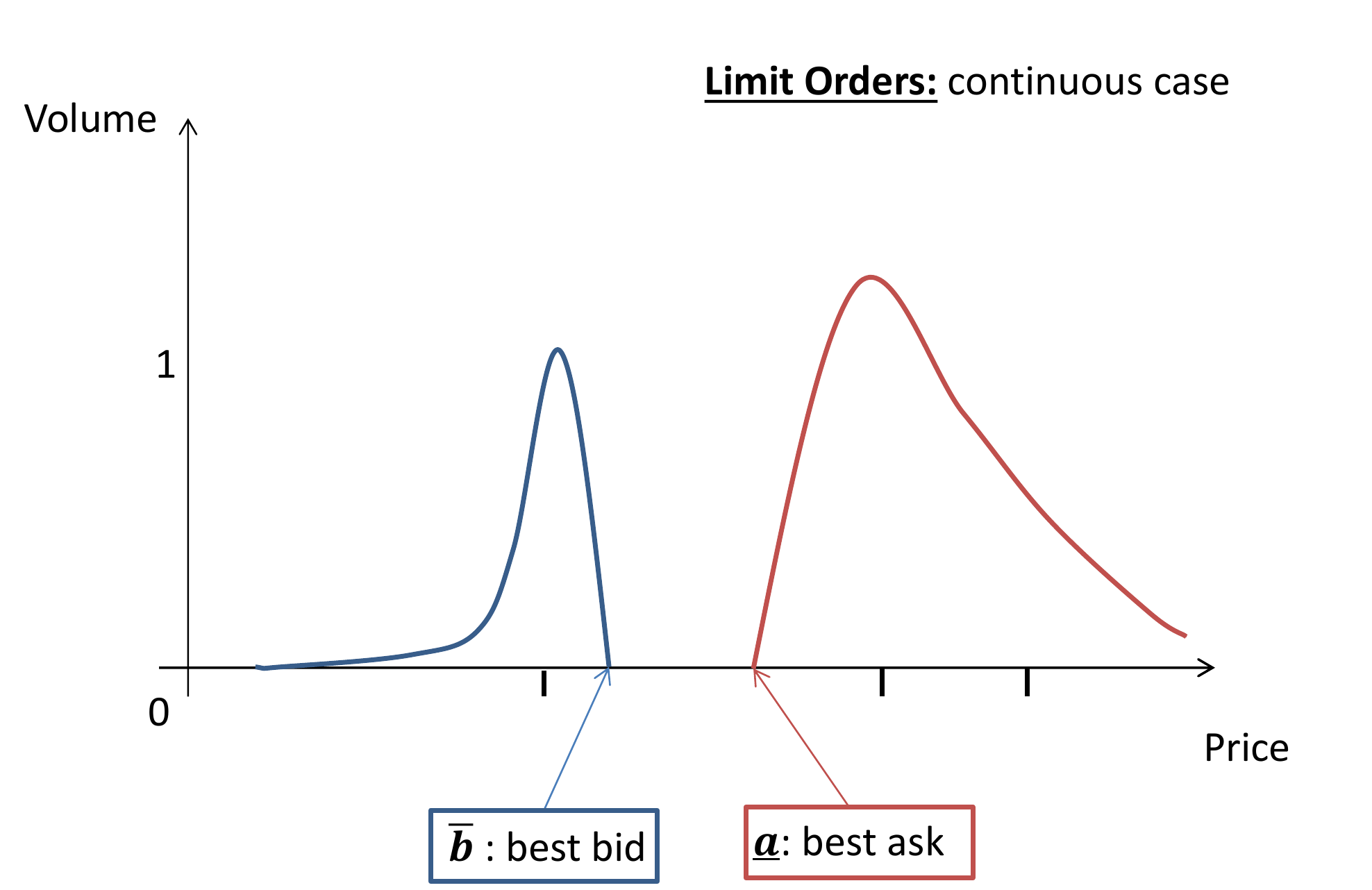}
\end{minipage}
\vskip 5pt

The control $(\beta,\alpha)$ of the liquidity taker represents his \emph{market orders}. A market order placed against the bids will cause all the bid orders above and including the price $\beta$ to be executed. For market orders against the ask, all the limit orders \emph{below} the level $\alpha$ will be executed. The limiting cases $\alpha = 0$ and $\beta = \infty$ correspond to 'empty' market orders that do not execute any limit orders. The execution of a market order leads to the following changes in cash and inventory: 

\begin{definition}[Execution of a market order]
Assume the order book is $(b,a)$ and that the liquidity taker chooses the pair $(\beta,\alpha)$. Then the change $\Delta L$ of \emph{inventory} triggered by the trade and the change $\Delta K$ in \emph{cash} that the \emph{liquidity provider} is subject to are defined by:
\begin{align}
\Delta L &= b[\beta,\infty) - a(0,\alpha] \label{atomic_asset_change}\\
\Delta K &= \int_{(0,\alpha]} x a(dx) -\int_{[\beta,\infty)} x b(dx) \label{atomic_cash_change}
\end{align}
\end{definition}

\vskip 5pt
\hskip -15pt
\begin{minipage}{.4\textwidth}
			\includegraphics[width=1\textwidth]{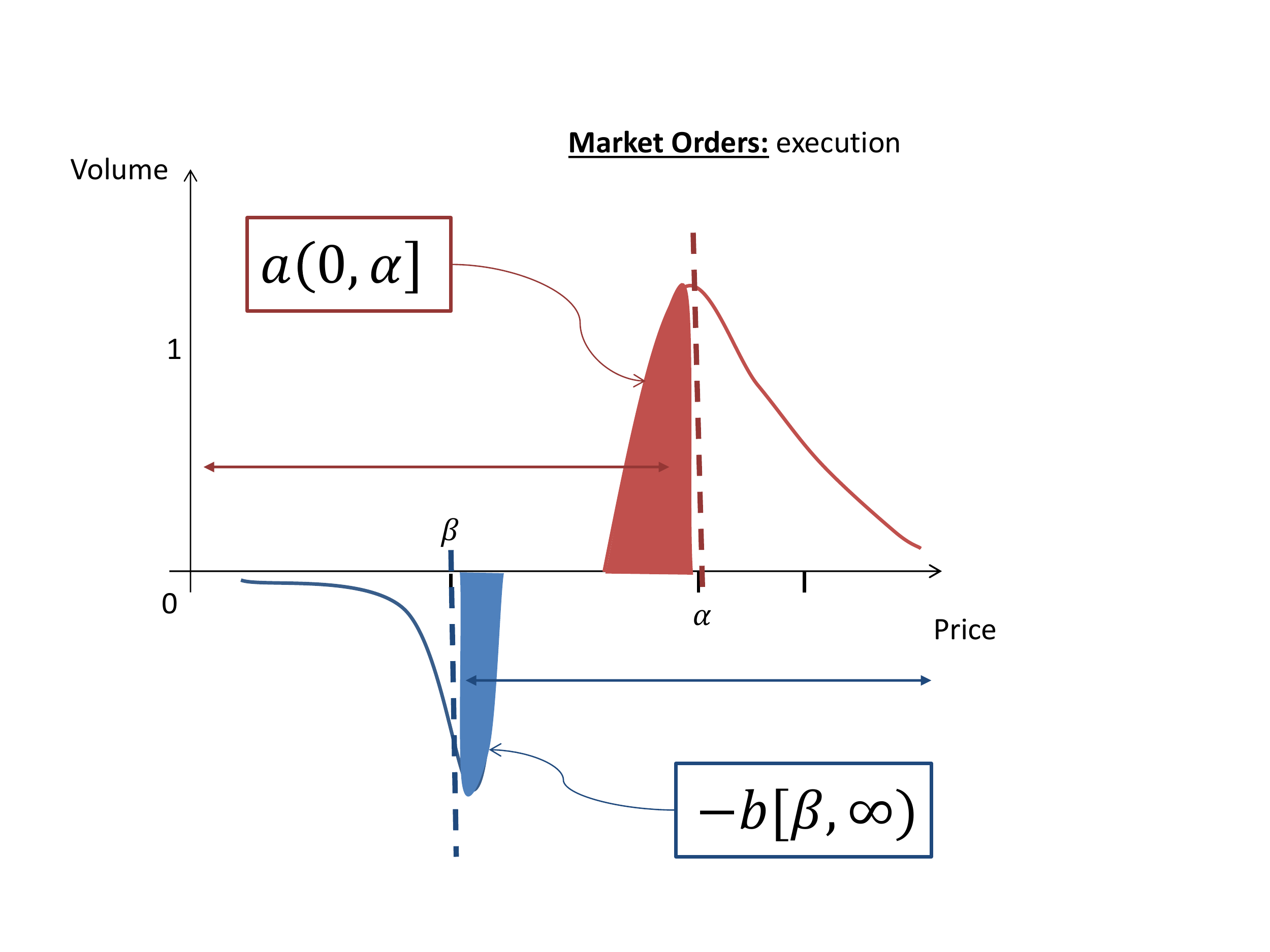}
\end{minipage}
\begin{minipage}{.6\textwidth}
For the justification of this formula let us first consider a single bid $b = \delta_p$. That is, the provider expresses interest in \emph{buying} one unit of the asset at the price $p$ or lower. A liquidity taker's market order to \emph{sell} will therefore execute the order if and only if its price level $\beta$ is \emph{smaller}. Should such an execution take place, the liquidity provider gains one unit of volume and loses $p$ units of cash. The above formula is then obtained by aggregating linearly the individual limit orders.
\end{minipage}
\vskip 5pt

The following assumptions will be used throughout the section.
\begin{hypothesis}
\label{assumption1}
The order books $(b,a)$ are such that $\bar{b}<\underline{a}$, that is, the bid-ask spread is always positive. We will say in this case that the order book \emph{exhibits no arbitrage}.
\end{hypothesis}

\begin{hypothesis}
\label{assumption2}
It is never optimal for the liquidity taker to buy and sell simultaneously. 
\end{hypothesis}
In particular, we can recode the liquidity taker's control by a single real number $\alpha$ by making him formally send the market orders $(\alpha, \alpha)$. Indeed, if $\alpha \in (\bar{b}, \underline{a})$ there is no trade, if $\alpha \ge \underline{a}$ a buy happens but no sell, and similarly for $\alpha \le \bar{b}$.

\subsubsection{\textbf{A probabilistic model for liquidity taker behavior}}\label{sub_sec_behavior}
We now provide a simple model for which Assumption \ref{assumption2} follows automatically from  Assumption \ref{assumption1}. 
Let $(\Omega, \F, \P)$ be a probability space modeling the beliefs of the liquidity taker. Let $p$ be a random variable representing the price at which the liquidity taker values the asset at a future time. We assume that the liquidity taker is risk-neutral under $\P$ in the sense that he maximizes his expected wealth after the trade, in other words he solves the optimization problem:
\begin{equation}
\max_{\beta,\alpha} \E \left[ - p \Delta L - \Delta K\right] \label{opt_taker}
\end{equation}

\begin{proposition}\label{taker_sol}
Let the order book $(a,b)$ be given. Then an optimal trade for the liquidity taker is given by
\begin{equation}
\beta = \alpha = \E[p].
\end{equation}
\end{proposition}
\begin{proof}
The liquidity taker looks for the supremum over $(0,\infty)\times (0,\infty)$ of the function
\begin{equation}
(\beta, \alpha) \longmapsto \int_{[\beta,\infty)} (x-\E[p]) b(dx) - \int_{(0,\alpha]} (x-\E[p]) a(dx) 
\end{equation}
This function decouples and we are left maximizing
\begin{equation}
\beta \longmapsto  \int_{[\beta,\infty)} (x - \E[p]) b(dx)
\end{equation}
which is non-decreasing on $(0,\E[p]]$ and non-increasing on $[\E[p],\infty)$. The same result holds for
\begin{equation}
\alpha \longmapsto  -\int_{(0,\alpha]} (x - \E[p]) a(dx)
\end{equation}
The supremum is attained for $\beta = \alpha = \E[p]$.
\end{proof}

\begin{remark}
While we do not have uniqueness of this maximum, all the other choices of optimum market orders will lead to exactly the same executions. Indeed, the function $\beta \longmapsto  \int_{[\beta,\infty)} (x - \E[p]) b(dx)$ and $\alpha \longmapsto  -\int_{(0,\alpha]} (x - \E[p]) a(dx)$  respectively do not have a strict maximum in $\E[p]$ iff $b$ and $a$ respectively put zero mass on some interval including $\E[p]$. Any market orders on this interval will lead to exactly the same cash and asset transfers and we can without loss of generality replace them by market orders at $\E[p]$. A similar argument can be made to rule out partial orders. In particular, we can summarize the taker's market orders by a single number $\alpha$.
\end{remark}

\begin{corollary}
Assume the order book $(b,a)$ exhibits \emph{no arbitrage}. Then it is never optimal for the taker to buy and sell simultaneously.
\end{corollary}
\begin{proof}
By the previous comment, we can summarize the market orders of a taker behaving optimally by a single real $\alpha$. The taker's buy and sell volumes are 
\begin{equation}
a[\underline{a},\alpha] \; \textit{and} \quad  b[\alpha,\bar{b}] 
\end{equation}
The no arbitrage property implies that these two terms cannot both be positive.
\end{proof}

\subsubsection{Alternative representation of the order book}
Even though the above representation of limit and market orders is clear, we still present an alternative description which only makes sense if no arbitrage is present on the market and Assumption \ref{assumption2} is verified.

The below definitions correspond to a very intuitive `graphic' approach. In the previous section, we have defined the order book as a pair of positive measures $(b,a)$. The no-arbitrage condition guarantees that these two measures have disjoint supports. One is therefore tempted to `glue' the two measures together into one. But in order to do that, we also need to keep track of where the offers starts and the bids stop. This is done in the following way.

\begin{definition}[Quoted price]
Let $(b,a)$ be an order book that does not exhibit arbitrage. Then we say that $p$ is a \emph{quoted price} of the order book if $p \in (\bar{b},\underline{a})$.
\end{definition}

Because the bid-ask spread is positive, there is not a unique quoted price. This is an unfortunate reality of high frequency markets, and we will only be able to mathematically resolve this difficulty in the limit where the bid-ask spread vanishes. Using a quoted price as a separation point between bid and ask limit orders, we can define:

\begin{definition}[Shape function]
Let $(b,a)$ be an order book that exhibits no arbitrage and $p$ be one of its quoted prices. Then define the order book's \emph{shape} function $\gamma: \R \longmapsto [0,\infty)$ to be
\begin{equation}
\gamma(u) = \int_0^u \left(a(0, p + x] -b[p + x,\infty)\right)dx. 
\end{equation}
In particular, $\gamma$ is convex, $\gamma(0)=0$ and $\gamma'(0) =0$. Moreover, $\gamma'$ is bounded and as a result, $\gamma$ has at most linear growth. \end{definition}

\begin{remark}
Notice that $\gamma''(\cdot + p) = b + a$  if both measures $b$ and $a$ have densities, or more generally, if we understand this equality in the sense of distributions.
\end{remark}

\begin{figure}[htbp]
	\centering
		\includegraphics[width=1.00\textwidth]{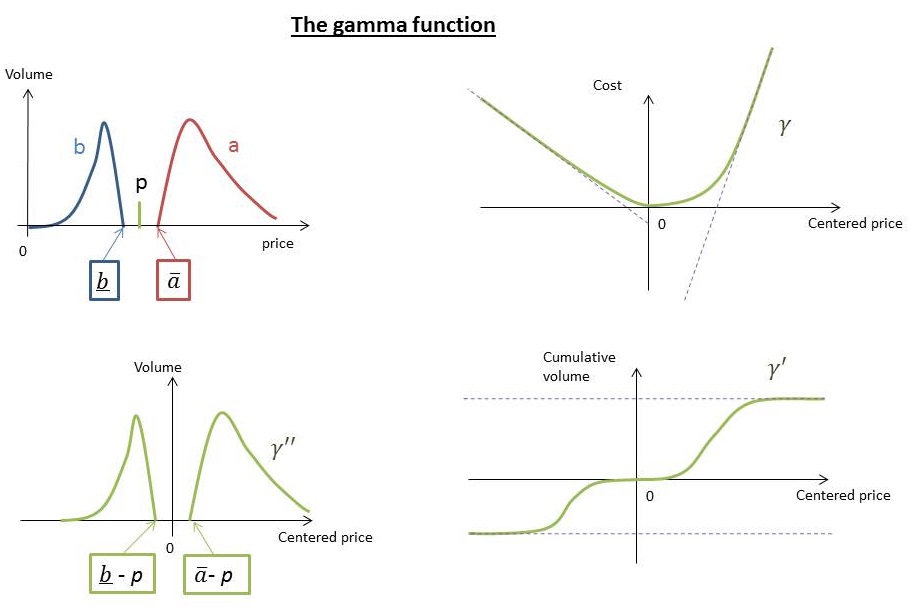}
	\label{fig:gamma}
\end{figure}

The following result recasts the trade equations in terms of the function $\gamma$.

\begin{proposition}\label{dual_formulas}
Let $(b,a)$ be an order book which exhibits no arbitrage, $p$ be one of its quoted prices and $\gamma$ the associated shape function. If $\alpha = u + p$ is the liquidity taker's market order, then we have
\begin{align}
\Delta L &= - \gamma'(u)\label{fo:deltaL}\\
\Delta K &= (u + p)\gamma'(u) - \gamma(u).\label{fo:deltaK}
\end{align}
\end{proposition}
\begin{proof}
The first identity is immediate from the definition of $\gamma$ and $\Delta L$:
\begin{align*}
\Delta L &= b[\alpha,\infty) - a(0,\alpha] \\
         &= - \gamma'(u)
\end{align*}
The second identity follows using integration by parts:
\begin{align*}
\Delta K &= \int_{(0,\alpha]} x a(dx) -\int_{[\alpha,\infty)} x b(dx) \\
         &= \alpha\, a(0,\alpha] - \int_{(0,\alpha]} a(0,x] dx - \alpha\, b[\alpha,\infty) + \int_{[\alpha,\infty)}  b[x,\infty)dx \\
				 &= \alpha \gamma'(u) - \gamma(u)
\end{align*}
\end{proof}

\begin{remark}
The liquidity provider's change in portfolio is captured by the pair $(\Delta L, \Delta K)$ comprising her inventory and cash positions. There are multiple ways to denote her change in \emph{wealth}. But if there is no price recovery, then the change in price of the asset after the transaction would be $\alpha- p$, and we have:
\begin{align*}
\Delta X &= (\alpha - p) \Delta L + \Delta K \\
         &= -\gamma(u)
\end{align*} 
for the transfer of wealth from the liquidity taker to the liquidity provider. Notice that we used \eqref{fo:deltaL} and \eqref{fo:deltaK} to deduce the second equality.
$\Delta X$ is always non-positive by construction of $\gamma$, and minimal at the quoted price used to define $\gamma$. As a result, the shape function can be seen as a measure of \emph{adverse selection} the liquidity provider is willing to incur at a given price level if price recovery were non-existent.
\end{remark}

To relate the transaction costs back to the traded volume without going through the transaction price $\alpha$, we use the following result:
\begin{proposition}{(Transaction costs)}
Define the \emph{transaction cost function} $c$ as the Legendre transform of $\gamma$:
\begin{equation}
c(l) = \sup_{u} \left(u l - \gamma(u)\right).
\end{equation}
Then we have: 
\begin{equation}
\Delta K =  - p \Delta L + c\left(\Delta L\right),
\end{equation}
and in particular, $c$ is convex and satisfies $c(0)=0$.
\end{proposition}
\begin{proof}
By the Fenchel identity, we have that
\begin{equation*}
u \gamma'(u)  = \gamma(u) + c(\gamma'(u))
\end{equation*}
and that $c'$ is the generalized inverse of $\gamma'$. Hence, as $\Delta L = \gamma'(u)$ we have that $u = c'(\Delta L)$ and
\begin{align*}
\Delta K &= - p \Delta L + u \gamma'(u)  - \gamma(u) \\
         &= - p \Delta L + c(\Delta L)
\end{align*}
\end{proof}

An order book $(b,a)$ which \emph{does not exhibit arbitrage} can therefore be represented by a pair $(p,\gamma)$ with $p$ a real and $\gamma$ a differentiable, convex function with linear growth satisfying $\gamma(0) = \gamma'(0) = 0$. Note that this representation in terms of \emph{quoted price} and order book \emph{shape} is \emph{not unique}, but leads to a completely equivalent description of trades and hence the same market model. 

Both representations have pros and cons and unfortunately, both will need to be juggled at different times of our analysis. The advantages of the original $(b,a)$ representation are: \emph{uniqueness} of the decomposition, ease to derive \emph{no-arbitrage relationships} and \emph{natural} interpretation of formulas. The alternative representation in terms of $(p,\gamma)$ is \emph{more tractable} and \emph{concise} as it involves a real number and a function rather than a pair of measures.

\subsubsection{Summary}\label{subsec_summary_micro}
For future reference, we summarize the different trade equations and market representations defined and derived in this section. 

The liquidity provider places limit orders. If the \emph{limit order book} formed that way presents no arbitrage, it will be represented either by a pair of measures $(b,a)$ or a couple $(p,\gamma)$ with $p$ a real number and $\gamma$ a differentiable, convex function with linear growth and $\gamma(0) = \gamma'(0) = 0$. Consistency equations between the two representations can be found above. 

We call $(b,a)$ the \emph{order book}, $p$ a \emph{quoted price} and $\gamma$ the \emph{shape} of the order book. The liquidity taker's \emph{market order} will be represented either by a real $\alpha$ representing a \emph{price}, or a real $u$ denoting a \emph{centered} price (shifted by the quoted price $p$). Both representations lead to the same trades.
\begin{align*}
\Delta L &= b[\alpha,\infty) - a(0,\alpha] \\
         &= - \gamma'(u)
\end{align*} 
is the change in inventory of the liquidity provider, while
\begin{align*}
\Delta K &= \int_{(0,\alpha]} x a(dx) -\int_{[\alpha,\infty)} x b(dx) \\
				 &= (u + p) \gamma'(u) - \gamma(u) \\
				 &= p \Delta L + c\left(\Delta L\right)
\end{align*}
is her change in cash position.

The market order corresponding to this trade can be recovered from the limit orders and the trade volume by the relationship
\begin{equation}
\alpha - p = c'(-\Delta L)
\end{equation}
and this is the price impact in the absence of price recovery.

\subsection{Discrete self-financing equation and other relationships}
We now give ourselves a discrete price process $p$ and provider inventory process $L$. Just as in the bid-ask spread case, three necessary conditions can be derived.
\subsubsection{Self-financing equation}
\begin{equation}
\Delta X = L \Delta p + c\left(\Delta L\right) + \Delta p \Delta L
\end{equation}
\subsubsection{Price impact}
\begin{equation}
\Delta p \Delta L \le 0
\end{equation}
\subsubsection{Price recovery}
\begin{equation}
|\Delta p| \le |c'(-\Delta L)|
\end{equation}

\subsection{Macroscopic limit}
The strategy in this section is identical to that of section \ref{sec:bid-ask}. We start off with the data of our problem in continuous time, discretize it to apply the discrete relationships derived earlier and finally take the diffusion limit to obtain our continuous time relationships.

\subsubsection{Approximation procedure}
Let $\left(\Omega, \F, \FF, \P\right)$ be a filtered probability space supporting a Wiener process $(W,W')$ with unspecified correlation structure. We consider a fixed time interval $[0,1]$ and give ourselves the following $\F$-adapted processes for the price and inventory of a provider:
\begin{equation}
	\begin{cases}
		p_t &= p_0 + \int_0^t \mu_u du + \int_0^t \sigma_u dW_u \\
		L_t &= L_0 + \int_0^t b_u du + \int_0^t l_u dW'_u
	\end{cases}
\end{equation}
where $p_0$ and $L_0$ are $\F_0$-measurable elements of $L^2$ and $\mu$, $\sigma$, $b$ and $l$ are $\F$-adapted and c\`adl\`ag processes. Finally, let $c : \Omega \times [0,1] \times \R \rightarrow \R^d$ be a random, $\F_t$-adapted function that is $C^0$ in $(t,l)$. Assume $c$ to be a.s. convex for all $t$, with a minimum at $c_t(0)=0$ and such that $c_t(l) < C l^2$ for some constant $C$. We denote by $\gamma_t$ its Legendre transform, which will represent the shape function of the order book as measured in tick size.

Let $\frac{1}{\sqrt{N}}$ be a vanishing tick size. Define the discretized price process as $p^N_n =p_{n/N}$ and likewise $L^N$.

\begin{figure}
	\centering
		\includegraphics[width=1.00\textwidth]{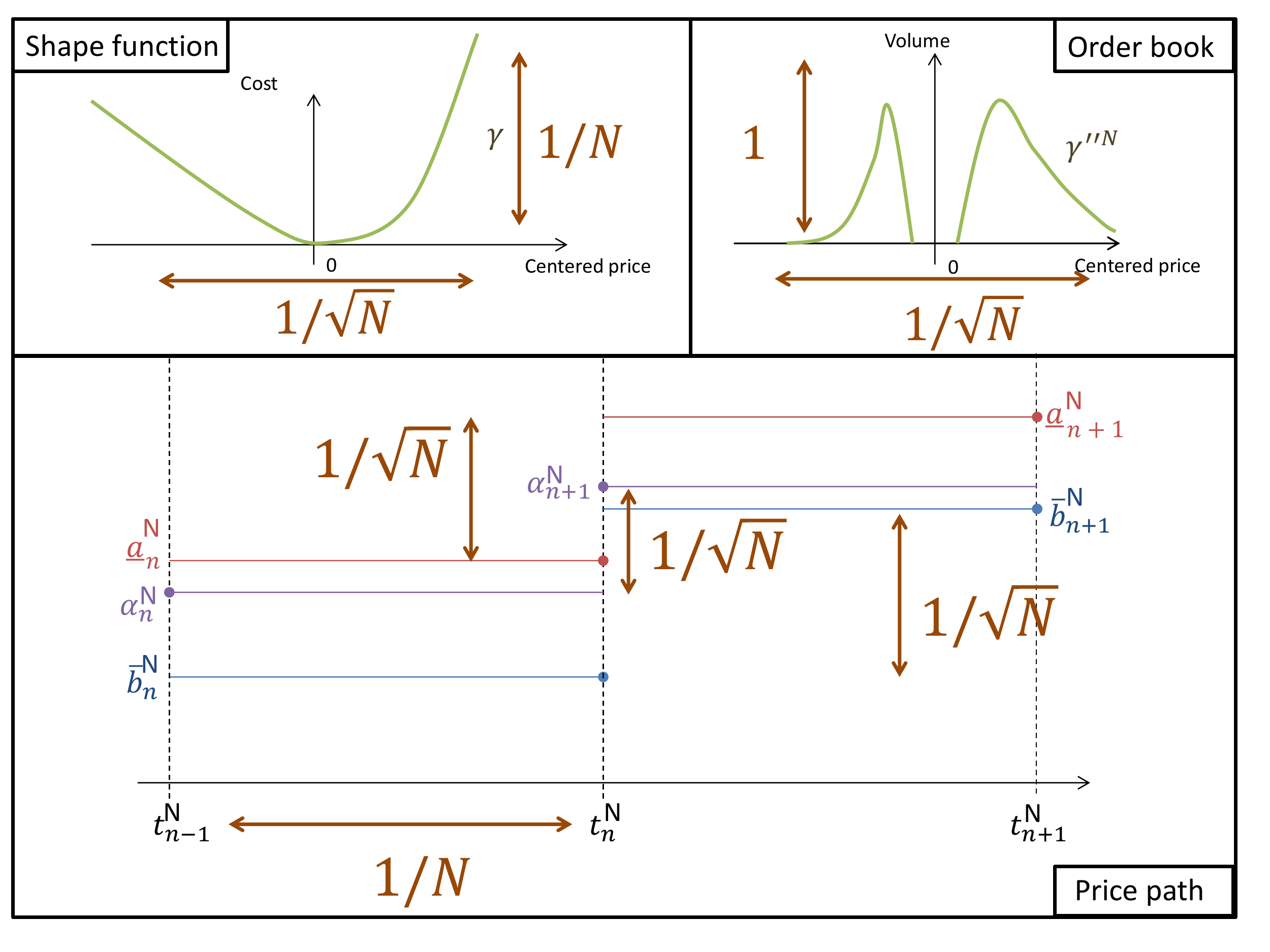}
	\caption{Renormalization of the model for the diffusion limit. Time is scaled by $1/N$, prices by $1/\sqrt{N}$ and volume by $1$ (unchanged). For example, the $y$-axis of $\gamma$ represents cost, that is $[\textit{volume}]\cdot[\textit{price}]^2$ which scales in $1/N$. The $x$-axis is expressed in prices and is scaled in $1/\sqrt{N}$, leading to the formula $\gamma^N(\cdot) = \gamma(\sqrt{N} \cdot)/N$.}
	\label{fig:Homogeinization} 
\end{figure}

We propose the following choice of renormalization for the order book.
\begin{equation}
\gamma^N_n(x) = \frac{1}{N} \gamma_{n/N}\left(\sqrt{N} x\right)
\end{equation}
This in particular implies 
\begin{equation}
c^N_n(l) = \frac{1}{N} c_{n/N}\left(\sqrt{N} l\right)
\end{equation}
This follows from the fact that $\gamma$ is defined in tick size and needs to be renormalized appropriately in the discrete approximation, where we want $\gamma^N$ to be expressed in absolute terms.

\subsubsection{Main result}
\begin{theorem}
The continuous time relationships between provider wealth $X$, inventory $L$, price $p$ and transaction costs $c$ are:
\begin{equation}
	\begin{cases}
		dX_t = L_t dp_t + \Phi_{l_t}(c_t) dt + d[L,p]_t \\
		d[L,p]_t \le 0 \\
		\sigma^2_t \le \Phi_{l_t}((c'_t)^2)
	\end{cases}
\end{equation}
where $X_t = \lim_{N\rightarrow\infty} X^N_{\lfloor Nt\rfloor}$ u.c.p.
\end{theorem}
\begin{proof}
Just as in the bid-ask spread case, the result to prove is the u.c.p. convergence of 
\begin{equation}
\frac{1}{N}\sum_{n=1}^{\lfloor tN\rfloor} c_{n/N}\left(\sqrt{N} \Delta_n L^N\right)
\end{equation}
and 
\begin{equation}
\frac{1}{N}\sum_{n=\lfloor t_1 N\rfloor}^{\lfloor t_2 N\rfloor} \left(\Delta_n p^N\right)^2 - \left(c'_{n/N}\left(\sqrt{N} \Delta_n L^N\right)\right)^2
\end{equation}
to the integrals
\begin{equation}
\int_0^t \Phi_{l_u}(c_u) du
\end{equation}
and
\begin{equation}
\int_{t_1}^{t_2} \left(\sigma^2_u - \Phi_{l_u}((c'_u)^2)\right) du
\end{equation}
This is a direct application of theorem \ref{thm_Jacod}.
\end{proof}

\section{\textbf{Naive supply and demand model}}
The aim of this section is to illustrate how a model for limit order fill rates and exact price recovery leads to models of the price as a function of trade volumes, or vice versa. This therefore models supply and demand in high frequency markets and closes the loop of our endeavor. However, we do not believe these models to be as accurate as the previously derived relationships and only use this section for illustrative purposes.

\subsection{Microscopic assumptions}
The proposed model is: \emph{perfect} fill rate and \emph{deterministic} price recovery.

\subsubsection{Disclaimer}
Unlike for the other microscopic relationships checked empirically in Section \ref{se:empirical}, the model considered now is \emph{not always consistent with empirical data}. Fill rates are definitely \emph{not} one, and price recovery is \emph{not} deterministic.

\subsubsection{Setup}
Let $(\Omega, \F, \P)$ be a probability space and $p$ and $L$ be two discrete time processes representing the price of the market and the inventory of a liquidity provider respectively. Let $\gamma$ be a $C^3$-function valued discrete time process representing our provider's shape function and $c$ its associated transaction costs.

\subsubsection{Additional relationships}
We translate 'perfect fill rate' and 'deterministic price recovery' by the following equation:
\begin{equation}
\Delta p = \lambda c'(-\Delta L)  \label{eq:supply_demand_p}
\end{equation}
or, equivalently
\begin{equation}
\Delta L = -\gamma'(\lambda^{-1} \Delta p) \label{eq:supply_demand_L}
\end{equation}
where $\lambda \in(0,1]$ is a real that encapsulates price recovery. The bigger $\lambda$, the smaller the price recovery.

\subsection{Macroscopic limit}
Equation \eqref{eq:supply_demand_p} allows a liquidity provider to derive the price from trade volumes and the order book, while equation \eqref{eq:supply_demand_L} derives the trade volumes from the prices and the order book. Both lead to the same consistency relationships between $p$, $L$ and $\gamma$ in the continuous limit.

\subsubsection{Main tool}
The proof method is based on another result from \cite{Jacod}. We first summarize the hypothesis and result before imposing them on the data of our problem.

Let $(\Omega, \F, \FF, \P)$ be a filtered probability space supporting an $1$-dimensional Wiener process $W$ and $Y$ be a $1$-dimensional  It\^o process of the form
\begin{equation}
Y_t = Y_0 + \int_0^t b_t dt + \int_0^t \sigma_t dW_t
\end{equation}
where we consider $t\in[0,1]$.

\begin{hypothesis}{(H)+ (K) from \cite{Jacod}}\label{J1}

Assume that $b_t$ and $\sigma_t$ are progressively measurable, $b_t$ is locally bounded and $\sigma_t$ is c\`adl\`ag.
\end{hypothesis}

Let now $F : \Omega \times [0,1] \times \R \rightarrow \R$ be a random, $\F_t$-adapted function that is $C^1$ in $y$ and $C^0$ in (t,y). We will shorten the notation to $y\mapsto F_t(y)$. Define the following assumption.

\begin{hypothesis}{(7.2.1), (10.3.2), (10.3.3), (10.3.4) and (10.3.7) from \cite{Jacod}}\label{J3}

Assume that a.s. for all $t$, $F_t$ is an \emph{odd} function in $y$.

Furthermore, assume there exists a function $g: \R \rightarrow \R$ with polynomial growth and a real $\beta > 1/2$ such that, for all $\omega \in \Omega$, $(t,s)\in[0,1]^2$ and $y \in \R$:
\begin{align*}
|F_t(y)| &\le g(y) \\
|F'_t(y)| &\le g(y) \\
|F_t(y) - F_s(y)| &\le g(y)|t-s|^\beta 
\end{align*}
\end{hypothesis}

Let us now state the new result from \cite{Jacod} we will use. 
\begin{theorem}{(10.3.2) from \cite{Jacod}}\label{thm_J2}
Assume \ref{J1} and \ref{J3}. Then there exists a very good filtered extension of the original space such that we have the following stable convergence in law as $N\rightarrow \infty$:
\begin{align*}
\frac{1}{\sqrt{N}}\sum_{n=1}^{\lfloor Nt \rfloor}& F_{n/N}\left({\sqrt{N}(X_{(n+1)/N} - X_{n/N})}\right) \rightarrow U_t
\end{align*}
where 
\begin{equation}
U_t = \int_0^t  b_s \Phi_{\sigma_s}\left(F'_s\right) ds  + \int_0^t \sqrt{\Phi_{\sigma_s}\left((F_s)^2\right)}	 dW'_s
\end{equation}
with $W'_t$ a  $d$-dimensional Wiener process such that
\begin{equation*}
[W',W]_t =  \int_0^t \frac{\Phi_{\sigma_s}\left( id \, F^k_s\right)}{\sigma_s \sqrt{\Phi_{\sigma_s}(F^k_s)^2}} ds
\end{equation*}
where $id$ is the identity function.
\end{theorem}

\subsubsection{Continuous time setup}
Let $(\Omega, \F, \FF, \P)$ be a filtered probability space supporting a Wiener process $W_t$. We will fix either an It\^o process
\begin{equation}
p_t = p_0 + \int_0^t \mu_s ds + \int_0^t \sigma_s dW_s
\end{equation}
for the price or 
\begin{equation}
L_t = L_0 + \int_0^t b_s ds + \int_0^t l_s dW_s
\end{equation}
for the inventory.

In addition to one of these processes, we also fix an order book shape process $\gamma_t$ and denote by $c_t$ the associated transaction cost process. 

Assume $L$ (respectively $p$) verifies Assumption \ref{J1} and $c$ (respectively $\gamma$) satisfies Assumption \ref{J3}. 

Just as previously, we define the discretized processes $L^N_n = L_{n/N}$ (respectively $p^N_n = p_{n/N}$) and $c^N_n(\cdot) = \frac{1}{N} c_{n/N}\left(\sqrt{N} \cdot\right)$ (respectively $\gamma^N_n(\cdot) = \frac{1}{N} \gamma_{n/N}\left(\sqrt{N} \cdot\right)$).

\subsubsection{Main result}
The main result is a straightforward application of Theorem \ref{thm_J2}. If we are given the inventory $L$ and transaction costs $c$ then we have:
\begin{theorem}
There exists a very good filtered extension of the original space such that we have the stable convergence in law $p^N_{\lfloor Nt\rfloor} \rightarrow p_t$ with
\begin{equation}
dp_t = -\lambda b_t \Phi_{l_t}\left(c''_t\right) dt + \lambda \sqrt{\Phi_{l_t}((c'_t)^2)} dW'_t
\end{equation}
where 
\begin{equation}
[W',W]_t = -\int_0^t \frac{\Phi_{l_s}\left(id \, c'_s\right)}{l_s \sqrt{\Phi_{l_s}((c'_s)^2)}}ds.
\end{equation}
In particular, 
\begin{equation}
d[p,L]_t = -\Phi_{l_t}\left(id \, c'_t\right)dt
\end{equation}
\end{theorem}
A completely equivalent result is obtained if the price $p$ and order book shape function $\gamma$ are given:
\begin{theorem}
There exists a very good filtered extension of the original space such that we have the stable convergence in law $L^N_{\lfloor Nt\rfloor} \rightarrow L_t$ with
\begin{equation}
dL_t =  -\mu_t \Phi_{\sigma_t}\left(\gamma''_t(\lambda^{-1} \cdot)\right) dt + \sqrt{\Phi_{\sigma_t}((\gamma'_t)^2(\lambda^{-1} \cdot))} dW'_t
\end{equation}
where
\begin{equation}
d[p,L]_t = -\Phi_{\sigma_s}\left(id \, \gamma'_t(\lambda^{-1} \cdot)\right)dt.
\end{equation}
\end{theorem}

\subsection{A special case}
A flat order book corresponds to $\gamma''_t = m_t$ for some adapted process $m$. While quite unrealistic, it is \emph{extremely} tractable and has been proposed and used in other models (\cite{Alfonsi, Wang}).

This corresponds to quadratic transaction costs and \emph{linear} price impact:
\begin{equation}
	\begin{cases}
		dp_t &= -\frac{\lambda}{m_t} dL_t \\
		dX_t &= L_t dp_t + \left(\frac{1}{2} - \lambda\right) \frac{l^2_t}{m_t} dt
	\end{cases}
\end{equation}

Note that the the sign of the effective transaction costs is that of $\frac{1}{2} - \lambda$. Indeed, in the self-financing case $\lambda = \frac{1}{2}$, price recovery and price impact perfectly cancel each other out. If $\lambda > \frac{1}{2}$, then the price impact of trades is stronger than the collected spread because of insufficient price recovery. Also, because of the uniform structure of the order book and perfect fill rate, the inventory of the provider is perfectly anti-correlated to the price.

\subsubsection{The worst case for providers}
As we have seen before, perfect anti-correlation is the worst case for the liquidity provider, making the uniform order book `the worse' shape from the liquidity provider's perspective. Amongst uniform order books, absence of price recovery, $\lambda =1$ is the wost case scenario.

A cute result is that if the liquidity provider provides constant liquidity ($m_t =1$) then we have the following identity between wealth and inventory: 
\begin{equation}
X_t = X_0 - L^2_t + L^2_0
\end{equation}
that is, even with the most naive strategy in the worst case scenario, the liquidity provider does not lose money if she manages her inventory. Symmetrically, one can show that, even in this best case scenario for liquidity takers, there are no round-trip statistical arbitrage opportunities due to price impact only.

\section{\textbf{Conclusions}}
In conclusion, the present paper identifies key features of \emph{high frequency} limit order book markets and derives corresponding \emph{necessary} conditions on self-financing portfolios for continuous-time models of such markets. These features are:
\begin{enumerate}
\item \emph{Non-smoothness} of inventories of high frequency traders and \emph{vanishing} bid-ask spread in high frequency markets.
\item \emph{Adverse selection}  as given by a negative quadratic covariation between price increments and change in provider inventory, which is a consequence of the \emph{price impact} of trades on such time-scales.
\item \emph{Price recovery} and the way it links the bid-ask spread and price volatility processes.
\item Generalized formula for the wealth process of \emph{a self-financing portfolio} when including price impact.
\item Applications to option hedging and portfolio optimization highlighting the differences between trades via market orders and limit orders, and the differences between liquidity providers and liquidity takers.
\end{enumerate}
These features were obtained by studying, both theoretically and empirically, high frequency market microstructure before summarizing it on a macroscopic level. As pointed out by \cite{O'Hara}, the crucial technical tool was the use of an \emph{event-based} clock. We hope further research will follow this method to uncover more effects of HFT on the broader financial system.

\appendix

\section{\textbf{Cross-sectional analysis}}
\label{se:tests}
The main empirical claim of the paper is the \emph{negative covariation} between liquidity provider inventory and the price process. This is one of many ways of identifying price impact, and is due to adverse selection of limit orders by liquidity takers. We wish to test this on a sample of stocks to identify when this relationship is verified, and when not. The data used in this appendix are $29$ large cap stocks using Nasdaq ITCH data on 18/04/13. Other days and stocks have been tested with similar results.

This test will come in three forms, from the most intuitive to the most sophisticated.

We first begin by listing for each of our $29$ stocks the proportion of trades not satisfying the property $\Delta L \Delta p \le 0$.

Then we plot the empirical quadratic covariations with confidence intervals constructed using the functional central limit theorem \cite{Yacine} for continuous It\^o processes.

\begin{figure}
	\centering
		\includegraphics[width=1.00\textwidth]{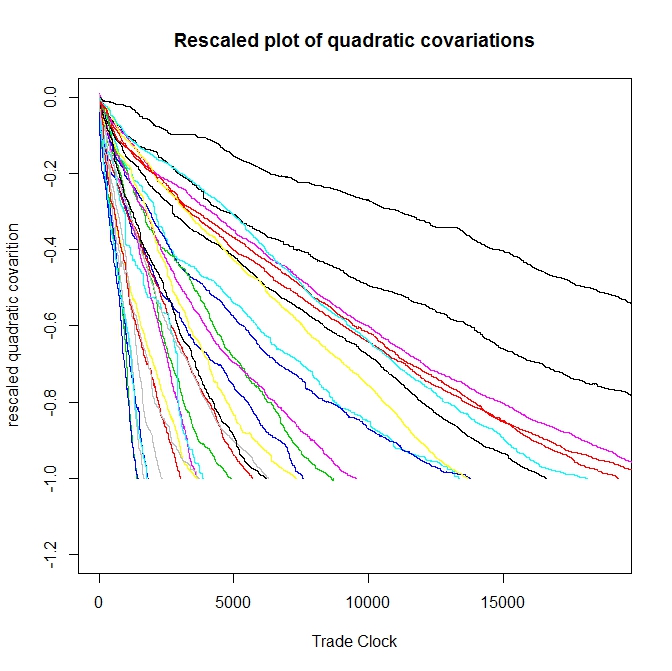}
	\caption{Empirical quadratic covariations (rescaled).}
	\label{fig:cross-section}
\end{figure}

Finally, we set up a rigorous statistical test based on the same functional central limit theorem. In the last case, we assume that we are given two continuous It\^o processes $L$ and $p$ such that:
\begin{equation}
	\begin{cases}
		dp_t &= \mu_t dt +\sigma_t dW_t \\
		dL_t &= b_t dt + l_t dW'_t
	\end{cases}
\end{equation}
with the quadratic covariation between $W_t$ and $W'_t$ being $\rho_t$. Assume furthermore that $\mu_t$ and $b_t$ to be locally bounded and that $\sigma_t$, $l_t$ and $\rho_t$ are c\`adl\`ag. 

If we then denote by $p^N$ and $L^N$ the discrete measurements of these processes on the uniform grid $\{1/N, 2/N, ..., 1\}$ then \cite{Yacine} tells us to consider the discrete processes:
\begin{equation}
	\begin{cases}
		C^N_t &= \sum_{n=1}^{\lfloor Nt\rfloor -1} \Delta_n p^N \Delta_n L^N \\
		V^N_t &= N \sum_{n=1}^{\lfloor Nt\rfloor -2} \left(  \left(\Delta_n p^N \Delta_{n+1} L^N\right)^2 + \Delta_n p^N \Delta_n L^N \Delta_{n+1} p^N \Delta_{n+1} L^N \right)
	\end{cases}
\end{equation}
and we have the functional central limit theorem
\begin{equation}
\mathcal{L}\left( \frac{C^N_t - [p,L]_t}{\sqrt{N^{-1} |V^N_t|}} \right) \rightarrow N(0,1)
\end{equation}

This allows the construction of confidence intervals for the quadratic covariation process. We also use this result to reject the following \emph{null hypothesis}:
\begin{hypothesis}
There exists $t \in [0,1]$ such that $\rho_t >0$.
\end{hypothesis}

by constructing confidence intervals for the quadratic covariation on small time intervals $[t_k,t_{k+1})$, we can compute rejection probabilities for the events $\rho_{t_k} >0$ for each $t_k$. By multiplying these rejection probabilities, we obtain the rejection probability for our overall null hypothesis. Our choice of time intervals $[t_k, t_{k+1}]$ is such that we have $100$ data points in each of these intervals.

Finally, we obtain  the tables:

\begin{table}[htbp]
	\centering
		\begin{tabular}{l|c|c|c|c|c }
		 Stock        & proba reject & nb false & nb trades & percent false  &recovery rejection\\
		\hline
 MSFT             &          0.7868301   &  19   &  27540   & 0.06899056 &      6.147422 \\
 KO               &          0.9876695   &  72   &  20362   & 0.3535998  &    13.932816 \\
 BA               &          0.9999383   & 222   &   4824   & 4.60199 &   24.212272 \\
 GPS              &          0.9999044   &  97   &   7378   & 1.314719  &  22.445107 \\
 GE               &          0.9991448   &   4   &  12969   & 0.03084278  &  6.847097 \\
 CS               &          0.8971721   & 132   &   3621   & 3.645402 &     37.448219 \\
 CPB              &          0.9421457   & 129   &   3578   & 3.605366  &     26.914477 \\
 BCS              &          0.9625842   &  43   &   1613   &  2.66584  &     27.774334 \\
 JNJ              &          0.9550316   & 152   &  16114   & 0.9432791 &     19.777833 \\
 UPS              &          0.9983282   & 237   &   5608   &  4.226106 &     30.117689 \\
 CLX              &          0.9563385   & 118   &   1381   & 8.544533 &  31.643736 \\
 T                &          0.9996831   &  27   &  13287   &  0.2032061  &     12.139685 \\
 DELL             &          0.9893074   &   1   &   3742   & 0.02672368 &      5.130946 \\
 XOM              &          0.9998707   & 340   &  20714   & 1.641402 &     19.276818 \\
 CAT              &          0.9814122   & 397   &  13456   & 2.950357 &     26.575505 \\
 COF              &          0.8973841   & 131   &   6103   & 2.146485  &     27.117811 \\
 AAPL             &          0.9999987   & 2347  &   46710  & 5.02462 &       9.648897 \\
 PG               &          0.9998587   & 189   &  18616   & 1.015256  &     18.038247 \\
 GOOG             &          0.9929220   & 609   &   8595   & 7.085515  &     15.602094 \\
 HSY              &          0.9615380   & 177   &   1807   &  9.795241  &     35.030437 \\
 WFC              &          0.9129410   &  13   &  17672   & 0.0735627 &     11.854912 \\ 
 DTV              &          0.6174753   & 117   &   9334   & 1.253482 &     21.952003 \\
 BBY              &          0.9999374   &  85   &   7181   & 1.183679  &     22.113912 \\
 MT               &          0.8870935   &  18   &   2273   & 0.791905 &     21.293445 \\
 GM               &          0.9774693   &  19   &   5963   & 0.3186316 &     18.279390 \\
 CL               &          0.9833529   & 187   &   3006   & 6.220892 &     24.550898 \\
 MA               &          0.9996761   & 113   &   1435   & 7.874564 &     18.048780 \\
 KSU              &          0.9945635   & 118   &   1756    & 6.719818 &      26.765376 \\
 GIS              &          0.9735843   &  68   &   3624   & 1.87638 &     22.323400 
		\end{tabular}
	\caption{Rejection probability of the null hypothesis, number of trades not satisfying our main inequality, total number of trades, percentage of trades not verifying our main inequality and percentage of trades not verifying our price recovery inequality. We also noted that all the lit trades across all the stocks happened at the best bid and best ask. Note that, of all our proposed relationships, the only weak one is price recovery, which is routinely violated.}
	\label{tab:table}
\end{table}

\end{document}